\documentclass[11pt]{article}
\usepackage{fullpage}

\usepackage{times}
\usepackage{comment,amsfonts,amssymb,amsmath,amsthm,graphicx,algorithm,algorithmic}
\newcommand{\commentout}[1]{}

\ifx\pdftexversion\undefined
\usepackage[colorlinks,linkcolor=black,filecolor=black,citecolor=black,urlco
lor=black,pdfstartview=FitH]{hyperref}
\else
\usepackage[colorlinks,linkcolor=blue,filecolor=blue,citecolor=blue,urlcolor
=blue,pdfstartview=FitH]{hyperref}
\fi

\newcommand{\alert}[1]{\textbf{\color{red}
[[[#1]]]}\marginpar{\textbf{\color{red}**}}\typeout{ALERT:
\the\inputlineno: #1}}

\def\MathF{\hbox{\rm I\kern-2pt F}}
\def\MathP{\hbox{\rm I\kern-2pt P}}
\def\MathR{\hbox{\rm I\kern-2pt R}}
\def\MathZ{\hbox{\sf Z\kern-4pt Z}}
\def\MathN{\hbox{\rm I\kern-2pt I\kern-3.1pt N}}
\def\MathC{\hbox{\rm \kern0.7pt\raise0.8pt\hbox{\footnotesize I}
\kern-4.2pt C}}
\def\MathQ{\hbox{\rm I\kern-6pt Q}}

\newcommand{\Prob}{\MathP}

\def\eps{\epsilon}
\def\tO{\tilde{O}}
\def\st{start}

\newcommand{\polylog}{{\rm polylog}}

\newcommand{\poly}{{\rm poly}}

\newcommand{\E}{{\mathbb{E}}}

\newcommand{\mommit}[1]{}
\newcommand{\namedref}[2]{\hyperref[#2]{#1~\ref*{#2}}}
\newcommand{\sectionref}[1]{\namedref{Section}{#1}}
\newcommand{\appendixref}[1]{\namedref{Appendix}{#1}}

\newcommand{\theoremref}[1]{\namedref{Theorem}{#1}}

\newcommand{\algref}[1]{\namedref{Algorithm}{#1}}
\newcommand{\claimref}[1]{\namedref{Claim}{#1}}
\newcommand{\lemmaref}[1]{\namedref{Lemma}{#1}}

\newcommand{\remarkref}[1]{\namedref{Remark}{#1}}
\newcommand{\corollaryref}[1]{\namedref{Corollary}{#1}}
\newcommand{\propertyref}[1]{\namedref{Property}{#1}}
\newtheorem{theorem}{Theorem}
\newtheorem{lemma}{Lemma}
\newtheorem{corollary}[lemma]{Corollary}
\newtheorem{remark}{Remark}
\newtheorem{claim}[lemma]{Claim}

\newtheorem{property}{Property}
\newtheorem{definition}{Definition}

\usepackage{pdfsync}
\usepackage{authblk}
\begin{document}

\title{On Efficient Distributed Construction of Near Optimal Routing Schemes\footnote{A preliminary version \cite{EN16a} of this paper was published in PODC'16.}}

\author[1]{Michael Elkin\thanks{This research was supported by the ISF grant No. (724/15).}}
\author[1]{Ofer Neiman\thanks{Supported in part by ISF grant No. (523/12) and by BSF grant No. 2015813.}}

\affil[1]{Department of Computer Science, Ben-Gurion University of the Negev,
Beer-Sheva, Israel. Email: \texttt{\{elkinm,neimano\}@cs.bgu.ac.il}}

\date{}
\maketitle

\begin{abstract}
Given a distributed network represented by a weighted undirected graph $G=(V,E)$ on $n$ vertices, and a parameter $k$, we devise a distributed algorithm that computes a routing scheme in $O(n^{1/2+1/k}+D)\cdot n^{o(1)}$ rounds, where $D$ is the hop-diameter of the network. Moreover, for odd $k$, the running time of our algorithm  is $O(n^{1/2 + 1/(2k)} + D) \cdot n^{o(1)}$. Our running time nearly matches the lower bound of $\tilde{\Omega}(n^{1/2}+D)$ rounds (which holds for any scheme with polynomial stretch). The routing tables are of size $\tilde{O}(n^{1/k})$, the labels are of size $O(k\log^2n)$, and every packet is routed on a path suffering stretch at most $4k-5+o(1)$. Our construction nearly matches the state-of-the-art for routing schemes built in a centralized sequential manner. The previous best algorithms for building routing tables in a distributed small messages model were by \cite[STOC 2013]{LP13} and \cite[PODC 2015]{LP15}. The former has similar properties but suffers from substantially larger routing tables of size $O(n^{1/2+1/k})$, while the latter has sub-optimal running time of $\tilde{O}(\min\{(nD)^{1/2}\cdot n^{1/k},n^{2/3+2/(3k)}+D\})$.
\end{abstract}


\section{Introduction}

A routing scheme in a distributed network is a mechanism that allows packets to be delivered from any node to any other node. The network is represented as a weighted undirected graph, and each node should be able to forward incoming data by using local information stored at the node, and the (short) packet's header. The local routing information is often referred to as a routing table. The routing scheme has two main phases: in the preprocessing phase, each node is assigned a routing table and a short label. In the routing phase, each node receiving a packet should make a local decision, based on its own routing table and the packet's header (which contains the label of the destination), to which neighbor forward the packet to.
The {\em stretch} of a routing scheme is the worst ratio between the length of a path on which a packet is routed, to the shortest possible path.

Designing efficient routing schemes is a central problem in the area of distributed networking, and was studied intensively \cite{PU89,ABLP90,C01,EGP03,GP03,AGM04,PU89,TZ01-spaa,C13}. The first general tradeoffs for this problem were given in pioneering works by \cite{PU89,ABLP90}.
In a seminal paper \cite{TZ01-spaa}, Thorup and Zwick presented the following compact routing scheme: Given a weighted graph $G$ on $n$ vertices and a parameter $k\ge 1$, the scheme has routing tables of size $\tilde{O}(n^{1/k})$,\footnote{The $\tilde{O}$ hides $\log^{O(1)}n$ factors.} labels of size $O(k\log n)$ and stretch $4k-5$. (Assuming that port numbers may be assigned by the routing process, otherwise the label size increases by a factor of $\log n$.)\footnote{They also presented stretch $2k-1$, assuming "handshaking": allowing the source and destination to communicate before the routing phase begins, but it is often desirable to avoid handshaking. Henceforth, we discuss only routing schemes that do not allow handshaking.}  The state-of-the-art is a scheme of \cite{C13}, which is based on \cite{TZ01-spaa}, and improves the stretch to $3.68k$.

All the results above assume that the preprocessing phase can be computed in a sequential centralized manner. However, as the problem of designing a compact routing scheme is inherently concerned with a distributed network, constructing the scheme efficiently in a distributed manner is a very natural direction. We focus on the standard CONGEST model \cite{P00}. In this model, every vertex initially knows only the edges touching it, and communication between vertices occurs in synchronous {\em rounds}. On every round, each vertex may send a small message to each of its neighbors. Every message takes a unit time to reach the neighbor, regardless of the edge weight. The time complexity is measured by the number of rounds it takes to complete a task (we assume local computation does not cost anything). Often the time depends on $n$, the number of vertices, and $D$, the {\em hop-diameter} of the graph. The hop-diameter is the maximum hop-distance between two vertices, where the hop-distance is the minimal number of edges on a path between the vertices (regardless of the weights). The hop-diameter is not to be confused with the {\em shortest path diameter} $S$, which is the maximal number of hops a shortest path uses (assuming shortest paths are unique). We always have $D\le S$, and typically $D$ is small while $S$ could be as large as $\Omega(n)$. We also assume, as common in the literature \cite{LP13,N14,KP98,GK13,HKN15}, that edge weights are integers and at most polynomial in $n$ (so that they could be sent in a single message).\footnote{We shall not consider {\em name-independent} routing, in which the label of a vertex is its ID, because \cite{LP13} showed a strong lower bound: any such scheme with stretch $\rho$ (even average stretch $\rho$) must take $\tilde{\Omega}(n/\rho^2)$ rounds to compute in this model.}

A rich research thread concerns with finding efficient distributed (approximation) algorithms for classical graph problems (e.g., minimum spanning tree, minimum cut, shortest paths), in sub-linear time \cite{GKP98,PR00,Elk06b,DHK12,HKN15}. There are several results obtaining running times of the form $\tilde{O}(\sqrt{n}+D)$, e.g. for MST, connectivity, minimum cut, approximate shortest path tree, etc. These results are often accompanied by a (nearly) matching lower bounds. The lower bound of \cite{DHK12}, based on \cite{PR00,E06},  implies that devising a routing scheme with any polynomial stretch, requires $\tilde{\Omega}(\sqrt{n}+D)$ rounds.

The first result on computing a routing scheme in a distributed manner within $o(n)$ rounds (for general graphs with $D=o(n)$), was shown by Lenzen and Patt-Shamir \cite{LP13}.\footnote{We remark that for the class of $k$-chordal graphs, \cite{NRS12} showed a construction of a routing scheme that could be computed efficiently in a distributed manner.} Their algorithm, given a graph on $n$ vertices and a parameter $k$, provides routing tables of size $\tilde{O}(n^{1/2+1/k})$, labels of size $O(\log n\cdot\log k)$, stretch at most $O(k\log k)$, and has a nearly optimal running time of $\tilde{O}(n^{1/2+1/k}+D)$ rounds. Note that the routing tables are of size $\Omega(\sqrt{n})$ for any value of $k$, which could be prohibitively large (the routing scheme of \cite{TZ01-spaa} supports stretch 3 with $\tilde{O}(\sqrt{n})$ table size). They also show implications for related problems, such as approximate diameter, generalized Steiner forest, and distance estimation. In a follow-up paper, \cite{LP15} showed how to improve the stretch of the above scheme to roughly $3k/2$ (for any $k$ divisible by 4). They also exhibited a different tradeoff, that overcame the issue of large routing tables. They devised an algorithm that produced routing tables of size $\tilde{O}(n^{1/k})$, labels of size $O(k\log^2n)$ and stretch $4k-3+o(1)$,\footnote{The paper \cite{LP15} claimed label size $O(k\log n)$, but in \cite{LP} it was communicated to us that the actual size is $O(k\log^2n)$.} but the number of rounds increases to $\tilde{O}(\min\{(nD)^{1/2}\cdot n^{1/k},n^{2/3+2/(3k)}+D\})$. Note that for moderately large hop-diameter $D\approx n^{1/3}$, the number of rounds is bounded by only $\approx n^{2/3}$ for any value of $k$. (They also show a variant where the number of rounds is $\tilde{O}(S+n^{1/k})$, but as was mentioned above, $S$ might be much larger than $D$.)

In the {\em distance estimation} problem (also known as sketching, or distance labeling), we wish to compute a small {\em sketch} for each vertex, so that given any two sketches, one can efficiently compute the (approximate) distance between the vertices. This problem was introduced in \cite{P00a}, who provided initial existential results. In \cite{SDP15}, a distributed (randomized) algorithm running in $\tilde{O}(S\cdot n^{1/k})$ rounds was shown, that computes sketches of size $O(kn^{1/k}\log n)$ with stretch at most $2k-1$. While this essentially matches the best sequential algorithm of \cite{TZ01}, the number of rounds could be $\Omega(n)$, even when $D$ is small. In \cite{LP13}, a running time of $\tilde{O}(n^{1/2+1/k}+D)$ rounds was presented, at the cost of significantly increasing the stretch to $O(k^2)$.\footnote{In fact, they showed a scheme in which it suffices to have a sketch of one vertex, and a $O(k\log n)$ size label of the other vertex, to derive the distance estimation. Our result has a similar property.} Izumi and Wattenhofer \cite{IW14} showed a lower bound of $n^{1/2 + \Omega(1/k)}$ rounds for this problem. In the Conclusion part of their paper \cite{IW14},  Izumi and Wattenhofer posed an open problem:
\\
\\
{\em ``An open problem related to our results is to find algorithms whose running time gets close to our lower bounds.''}

\paragraph{Our contribution.} We devise a randomized distributed algorithm running in\\ $(n^{1/2+1/k}+D)\cdot \min\{(\log n)^{O(k)},2^{\tilde{O}(\sqrt{\log n})}\}$ rounds, 
that with high probability, computes a compact routing scheme with routing tables of size $O(n^{1/k}\log^2n)$, labels of size $O(k\log^2n)$, and stretch at most $4k-5+o(1)$.
Moroever, for odd $k$, the running time of our algorithm is $(n^{1/2+1/(2k)}+D)\cdot \min\{(\log n)^{O(k)},2^{\tilde{O}(\sqrt{\log n})}\}$.
Note that our result nearly matches the construction of \cite{TZ01-spaa}, up to logarithmic terms in the size and $o(1)$ additive term in the stretch. This is even though the latter is computed in a sequential centralized manner. Observe that our running time nearly matches the lower bound of \cite{DHK12}, and is substantially better than that of \cite{LP15} whenever $D\ge n^{\Omega(1)}$ (which achieved similar size-stretch tradeoff). The previous result obtaining near optimal running time \cite{LP13}, suffers from excessive routing table size.

As a corollary, we show a distance estimation scheme, that can be computed in a distributed manner in $(n^{1/2+1/k}+D)\cdot \min\{(\log n)^{O(k)},2^{\tilde{O}(\sqrt{\log n})}\}$ rounds for even $k$, and for odd $k$ in
$(n^{1/2+1/(2k)}+D)\cdot \min\{(\log n)^{O(k)},2^{\tilde{O}(\sqrt{\log n})}\}$ rounds,
 providing sketches of size $O(n^{1/k}\log n)$ with stretch $2k-1+o(1)$. Each distance estimation takes only $O(k)$ time. Our result combines the improved running time of \cite{LP13} (up to lower order terms), with the near optimal size-stretch tradeoff of \cite{SDP15}. Moreover, our bound for the running time of distance estimation scheme nearly matches the lower bound $n^{1/2 + \Omega(1/k)}$ of Izumi and Wattenhofer \cite{IW14}, addressing their open problem. See Table \ref{fig:table} for a concise summary of previous and our results.

We note that to the best of our knowledge, all existing routing schemes \cite{PU89,ABLP90,TZ01,AGM04,C13,LP}, as well as the routing scheme that we present in this paper, enable distance estimation, i.e., given routing tables and labels of a pair $u,v$ of vertices, one can compute (without communication) a distance estimate $\hat{d}(u,v)$, which approximates the actual distance $d_G(u,v)$ between $u$ and $v$ up to  the stretch factor of the routing scheme. All routing schemes of this type require, by the lower bound of \cite{IW14}, at least $n^{1/2 + \Omega(1/k)}$ rounds to compute.


When preparing this submission, we learnt that concurrently and independently of us \cite{LPP16} came up with
a distributed algorithm running in $(n^{1/2+1/k}+D)\cdot 2^{\tilde{O}(\sqrt{\log n})}$ rounds, that with high probability, computes a routing scheme with routing tables of size $\tilde{O}(n^{1/k})$, labels of size $O(k\log^2n)$, and stretch at most $4k-3+o(1)$. Their result has slightly worse stretch, and a larger number of rounds whenever $k<\sqrt{\log n/\log\log n}$, or if $k$ is odd.
\begin{table}\label{fig:table}
\begin{center}
\begin{tabular}{|c|c|c|c|c|}
	\hline
                &  Number of Rounds   &   Table size & Label size & Stretch  \\
	\hline
\cite{TZ01-spaa,C13}   & $O(m)$ & $\tilde{O}(n^{\frac{1}{k}})$ & $O(k\log n)$ & $3.68k$  \\
\cite{LP15} & $\tilde{O}(S+n^{\frac{1}{k}})$ & $\tilde{O}(n^{\frac{1}{k}})$ & $O(k\log n)$ & $4k-3$\\
\cite{LP13,LP15}  & $\tilde{O}(n^{\frac{1}{2}+\frac{1}{4k}}+D)$ & $\tilde{O}(n^{\frac{1}{2}+\frac{1}{4k}})$ & $O(\log n)$ & $6k-1+o(1)$  \\
\cite{LP15} & $\tilde{O}(\min\{(nD)^{\frac{1}{2}}\cdot n^{\frac{1}{k}},n^{\frac{2}{3}+\frac{2}{3k}}+D\})$ & $\tilde{O}(n^{\frac{1}{k}})$ & $O(k\log^2n)$ & $4k-3+o(1)$\\
\hline
{\bf This paper},  & $(n^{\frac{1}{2}+\frac{1}{k}}+D)\!\cdot \!\min\{(\log n)^{O(k)},2^{\tilde{O}(\sqrt{\log n})}\}$ & $\tilde{O}(n^{\frac{1}{k}})$ & $O(k\log^2n)$ & $4k-5+o(1)$\\
even $k$ &&&& \\
{\bf This paper},  & $(n^{\frac{1}{2}+\frac{1}{2k}}+D)\!\cdot\! \min\{(\log n)^{O(k)},2^{\tilde{O}(\sqrt{\log n})}\}$ & $\tilde{O}(n^{\frac{1}{k}})$ & $O(k\log^2n)$ & $4k-5+o(1)$\\
odd $k$ &&&& \\
	\hline
\end{tabular}
\end{center}
\caption{Comparison of compact routing schemes for graphs with $n$ vertices, $m$ edges, hop-diameter $D$, and shortest path diameter $S$.}
\end{table}

\subsection{Overview of Techniques}

Let us first briefly sketch the Thorup-Zwick construction of a routing scheme. First they designed a routing scheme for trees, with  routing tables
of constant size and logarithmic label size. (Throughout the paper, the size is measured in RAM words, i.e., each word is of size $O(\log n)$.) For a general graph $G=(V,E)$ on $n$ vertices, they randomly sample a collection of sets $V=A_0\supseteq A_1\dots\supseteq A_k=\emptyset$, where for each $0<i<k$, each vertex in $A_{i-1}$ is chosen independently to be in $A_i$ with probability $n^{-1/k}$. 
The {\em cluster} of a vertex $u\in A_i\setminus A_{i+1}$ is defined as
\begin{equation}\label{eq:clustr}
C(u)=\{v\in V~:~ d_G(u,v)<d_G(v,A_{i+1})\}~.
\end{equation}
They proved that each cluster $C(x)$ can be viewed as a tree rooted at $x$, and showed an efficient procedure that given a pair $u,v\in V$, finds a vertex $x$ so that routing in the tree $C(x)$ has small stretch. So each vertex $u$ maintains in its routing table the routing information for all trees $C(x)$ containing it, while the label of $u$ consists of the tree-labels for a few special trees. They also show that (with high probability) every vertex is contained in at most $\tilde{O}(n^{1/k})$ trees.

The first difficulty we must deal with is that the routing scheme of Thorup-Zwick for a (single) tree could take a linear number of rounds to construct. We thus develop a variation on that scheme, that can be implemented efficiently in a distributed network. The basic idea is inspired by \cite{KP98} (and also used in \cite{N14}), which is to select $\approx\sqrt{n}$ vertices that partition the tree into bounded depth subtrees. We then apply the TZ-scheme locally in every subtree. The subtler part is to design a global routing scheme for the virtual tree\footnote{By a virtual tree we mean a tree whose edges are not present in the network.} induced on the sampled vertices, which must incorporate the local routing information.

\paragraph{Approximate Clusters.}

Once we have a distributed algorithm for routing in trees, we set off to apply the TZ-scheme for general graphs.
Unfortunately, it is not known how to compute the exact clusters efficiently in a distributed manner. 
In order to circumvent this barrier, we introduce the notion of {\em approximate clusters}. An approximate cluster is a subset of a cluster, that may exclude vertices that are "near" the boundary. (Slightly more formally, we may omit vertices for which the inequality \eqref{eq:clustr} becomes false if we multiply the left hand side by a $1+\epsilon$ factor, for a small $\epsilon>0$.) Our main technical contributions are: exhibiting a procedure that computes these approximate clusters, and showing that these approximate clusters are sufficient for constructing a routing scheme, with nearly matching size and stretch as in \cite{TZ01-spaa}.

The construction of clusters $C(u)$ for $u\in A_i\setminus A_{i+1}$,  where $i< k/2$, can be done in a straightforward manner (within the allotted number of rounds), since the depth of the corresponding tree is $\tilde{O}(\sqrt{n})$ with high probability, and since the {\em overlap} (the number of clusters containing a fixed vertex) is only $\tilde{O}(n^{1/k})$. The main challenge is computing the approximate clusters in the large scales, for $i\ge k/2$. To this end, we employ several tools. The first is {\em approximate multi-source hop-bounded distance computation}, which appeared recently in \cite{N14} (a certain variant of it  appeared also in \cite{LP13b}).
This enables us to compute approximations for $B$-hops shortest paths (paths that use at most $B$ edges), from a given $m$ sources to every vertex, in $\tilde{O}(B+m+D)$ rounds. The second tool we use is {\em hopsets}. The notion of hopsets was introduced by \cite{C00} in the context of parallel approximate shortest path algorithms, and it has found applications in dynamic, streaming and distributed settings as well \cite{B09,HKN14,HKN15}. A $(\beta,\epsilon)$-hopset is a (small) set of edges $F$, so that every shortest path has a corresponding $\beta$-hops path, whose weight is at most $1+\epsilon$ larger.

We compute the approximate clusters in the large scales as follows. First we sample $\approx\sqrt{n}$ vertices (those in $A_{k/2}$), and compute approximate $\sqrt{n}$-hops shortest paths from all the sampled vertices. Next we apply a $(\beta,\epsilon)$-hopset on the graph induced by these sampled vertices, where $\beta\le 2^{\tilde{O}(\sqrt{\log n})}$ and $\epsilon\approx 1/k^4$.
(A pair of sampled vertices is connected in this graph if and only if one is reachable from the other via an approximate $\sqrt{n}$-hop-bounded shortest path.)
An efficient distributed algorithm to construct such hopsets is given by \cite{HKN15,EN16}. We shall use the construction of \cite{EN16}, since it facilitates much smaller $\beta$, whenever $k$ is small.
(There are also some additional properties of hopsets from \cite{EN16}, that make them more convenient in the context of routing. See Section \ref{sec:prel}.)
This enables us to compute the approximate clusters on the sampled vertices, since we need only $\beta$ steps of exploration from each source $u$, using again that the overlap is small. Finally,  we extend each approximate cluster to the other vertices, by initiating an exploration from each sampled vertex to hop-distance $\approx\sqrt{n}$ in the original graph (in fact, one can use the multi-source hop-bounded distance computation of \cite{N14}). The correctness follows since with high probability, every vertex that should be included in some approximate cluster $\tilde{C}(u)$, has either $u$ or a sampled vertex within $\approx\sqrt{n}$ hops on the shortest path to it. The thresholds for entering an approximate cluster must be set carefully, so that every vertex on that shortest path will also join $\tilde{C}(u)$, in order to guarantee that the trees will indeed be connected (which is clearly crucial for routing), and on the other hand, to make sure that no vertex participates in too many trees. Unlike the exact TZ clusters, approximate clusters generally do not have to be connected.

The fact that our clusters are only approximate induces increased stretch. The analysis is similar to that of \cite{TZ01}, which consists of $k$ iterations of searching for the "right" tree. We must pay a factor of $1+O(\epsilon)$ in every one of these iterations, but fortunately, the hopset construction allows us to take sufficiently small $\epsilon$, so that all the additional stretch accumulates to an additive $o(1)$.

From a high level, our approach is similar to those of \cite{LP13,LP15}. In \cite{LP15}, they also use a variant of the TZ-routing scheme, which allows small errors in the distance estimations. The main difference is in handling the large scales. In \cite{LP13}, the idea was to build a spanner on a sample of $\approx\sqrt{n}$ vertices, which  reduces the number of edges. So a routing scheme can be efficiently computed on the spanner, and then extended to the entire graph. This approach inherently suffers from large storage requirement, since every vertex needs to know all the spanner edges. In \cite{LP15} the idea was to "delay" the start of large scales from $k/2$ to roughly $l_0=(k/2)\cdot(1+ \log D/\log n)$. Then they apply a distance estimation on the sampled vertices at scale $l_0$ (those in $A_{l_0}$) to construct the routing tables for all higher scales, and extend these to the remainder of the graph. However, the exploration in the graph on $A_{l_0}$ may need to be of $\approx n^{1-l_0/k}$ hops, which induces a factor of $D\cdot n^{1-l_0/k}=(nD)^{1/2}$ to the number of rounds.
The use of hopsets allows us to avoid the large memory requirement, since the routing is oblivious to the hopset, while significantly shortening the exploration range. Since the exploration range is proportional to the running time, the latter also decreases.


\subsection{Organization}

After stating in \sectionref{sec:prel} some of the tools we shall apply, in \sectionref{sec:route} we describe the notion of approximate clusters, and show how to compute these efficiently in a distributed manner. Then in \sectionref{sec:ana}, we demonstrate how these approximate clusters could be used for a routing scheme in general graphs. In \sectionref{sec:sketch} we show the distance estimation scheme. Finally, in \sectionref{sec:tree} we show our distributed tree routing.

\section{Preliminaries}\label{sec:prel}

Let $G=(V,E,w)$ be a weighted graph on $n$ vertices. We assume that $w:E\to\{1,\dots,\poly(n)\}$ (without this assumption, there will be a logarithmic dependence on the aspect ratio in the data structures' size and running times). Let $D$ be the {\em hop-diameter} of $G$, that is, the diameter of $G$ if all weights were 1. Denote by $d_G$ the shortest path metric on $G$. 
Let $d_G^{(t)}$ be the {\em $t$-hops} shortest path distance (abusing notation, since this is not a metric). That is, $d_G^{(t)}(u,v)$ is the shortest length of a path from $u$ to $v$, that has at most $t$ edges (set $d_G^{(t)}(u,v)=\infty$ if every path from $u$ to $v$ has more than $t$ edges). For each $u,v\in V$, define $h_G(u,v)$ as the number of hops on the shortest path in $G$ between $u$ and $v$. We shall always use this notation with respect to the input graph $G$, and thus will omit the subscript.
A (dominating) {\em virtual graph} on $G$ is a graph $G'=(V',E',w')$ with $V'\subseteq V$, and for every $u,v\in V'$ we have that $d_{G'}(u,v)\ge d_G(u,v)$. Every vertex in $V'$ should know all the edges of $E'$ touching it.
The following lemma formalizes the broadcast ability of a distributed network (see, e.g., \cite{P00}).
\begin{lemma}\label{lem:pipe}
Suppose every $v\in V$ holds $m_v$ messages, each of $O(1)$ words, for a total of $M=\sum_{v\in V}m_v$. Then all vertices can receive all the messages within $O(M+D)$ rounds.
\end{lemma}


\subsection{Tools}

We will make use of the following theorem due to \cite[Theorem 3.6]{N14}, which shows how to compute hop-bounded distances from a given set of sources, efficiently in a distributed manner.
\begin{theorem}[\cite{N14}]\label{thm:Nanongkai}
Given a weighted graph $G=(V,E,w)$ of hop-diameter $D$, a set $V'\subseteq V$, and parameters $B\ge 1$ and $0<\epsilon<1$, there is a (randomized) distributed algorithm that w.h.p runs in $\tilde{O}(|V'|+B+D)/\epsilon$ rounds, so that every $u\in V$ will know values $\{d_{uv}\}_{v\in V'}$ satisfying\footnote{The computed values are symmetric, that is, $d_{uv}=d_{vu}$ whenever $u,v\in V'$.}
\begin{equation}\label{eq:duv}
d_G^{(B)}(u,v)\le d_{uv}\le (1+\epsilon)d_G^{(B)}(u,v)~,
\end{equation}
\end{theorem}
\begin{remark}\label{rem:parents}
While not explicitly stated in \cite{N14}, the proof also provides that each $u\in V$ knows, for every $v\in V'$, a vertex $p=p_v(u)$ which is a neighbor of $u$ satisfying
\begin{equation}\label{eq:p-u}
d_{uv}\ge w(u,p)+d_{pv}~.
\end{equation}
\end{remark}
\paragraph{Hopsets.}
The following notion of hopsets was introduced by \cite{C00}.
\begin{definition}[Hopsets]
A set of (weighted) edges $F$ is a $(\beta,\epsilon)$-hopset for a graph $G=(V,E)$, if in the graph $H=(V,E\cup F)$, for every $u,v\in V$,
\begin{equation}\label{eq:hopset1}
d_G(u,v)\le d_H(u,v)\le d_H^{(\beta)}(u,v)\le(1+\epsilon)d_G(u,v)~.
\end{equation}
\end{definition}
We will need the following {\em path-reporting} property from our hopset. This property will be crucial for the connectivity of the trees corresponding to the approximate clusters.
\begin{property}\label{prop:hop}
A hopset $F$ for a graph $G$ is called {\em path-reporting}, if for every hopset edge $(u,v)\in F$ of weight $b$, there exists a corresponding path $P$ in $G$ between $u$ and $v$ of length $b$. Furthermore, every vertex $x$ on $P$ knows $d_P(x,u)$ and $d_P(x,v)$, and its neighbors on $P$.
\end{property}

The following result is from 
\cite{EN16}, which provides a path-reporting hopset. We remark that the original hopset construction of \cite{C00} could be made path-reporting. Also, in \cite[Theorem 4.10]{HKN15}, a distributed algorithm constructing a hopset is provided, which possibly could be made path-reporting, however, it inherently cannot provide a better hopbound than $2^{\tilde{O}(\sqrt{\log n})}$.
\begin{theorem}[\cite{EN16}]\label{thm:hopset}
Let $G$ be a weighted graph on $n$ vertices with hop-diameter $D$, let $0<\epsilon<1$, and let $G'$ be a virtual graph on $G$ with $m$ vertices. Let $0<\rho<1/2$ be a parameter, and write $\beta=\left(\frac{\log m}{\epsilon\cdot\rho}\right)^{O(1/\rho)}$. Then there is a randomized distributed algorithm that w.h.p computes in $\tilde{O}(m^{1+\rho}+D)\cdot \beta^2$ rounds, a path-reporting $\left(\beta,\epsilon\right)$-hopset $F$ for $G'$. 
\end{theorem}
We remark that in many applications (see, e.g., applications in \cite{C00,EN16}) the size of the hopset is important. However, here we only care about the size to the extent that it affects the number of rounds required to compute the hopset.

\paragraph{Approximate Shortest Path Tree (SPT).}
Recently, \cite{HKN15} obtained an efficient distributed algorithm for computing an approximate SPT, which we shall use. Let us first define the problem formally.
Let $G=(V,E,w)$ be a weighted graph. Given a set of vertices $A\subseteq V$, computing an $(1+\epsilon)$-approximate SPT rooted at $A$, means that every vertex $u\in V$ will know a value $\hat{d}(u)$ satisfying
\begin{equation}\label{eq:stretc}
d_G(u,A)\le\hat{d}(u)\le(1+\epsilon)d_G(u,A)~,
\end{equation}
and that $u$ will know a vertex $\hat{z}(u)\in A$ so that $d_G(u,\hat{z}(u))\le\hat{d}(u)$. The following theorem is a slight variation on a theorem shown in \cite{HKN15}. Here we use the hopsets of \cite{EN16} for an improved running time.

\begin{theorem}\label{thm:SPT}
Let $G=(V,E,w)$ be a weighted graph on $n$ vertices with hop-diameter $D$. Given a set $A\subseteq V$ of size $|A|\le 2\sqrt{n}\ln n$, and $\frac{1}{\polylog~n}<\epsilon<1$, there is a distributed algorithm that computes an $(1+\epsilon)$-approximate SPT rooted at $A$ in $(n^{1/2+1/(2k)}+D)\cdot \min\{(\log n)^{O(k)},2^{\tilde{O}(\sqrt{\log n})}\}$ rounds.
\end{theorem}
We defer the proof to \appendixref{app:SPT}.

\section{Distributed Routing Scheme}\label{sec:route}

In this section we define the notions of approximate pivots and approximate clusters, and describe an efficient distributed algorithm that computes these. Let us first recall the basic definitions from \cite{TZ01}.

Let $G=(V,E,w)$ be a weighted graph, fix $k\ge 1$. Sample a collection of sets $V=A_0\supseteq A_1\dots\supseteq A_k=\emptyset$, where for each $0<i<k$, each vertex in $A_{i-1}$ is chosen independently to be in $A_i$ with probability $n^{-1/k}$. A point $z\in A_i$ is called an $i$-pivot of $v$, if $d_G(v,z)=d_G(v,A_i)$. The cluster of a vertex $u\in A_i\setminus A_{i+1}$ is defined as
\begin{equation}\label{eq:cluster}
C(u)=\{v\in V~:~ d_G(u,v)<d_G(v,A_{i+1})\}~.
\end{equation}
We quote a claim from \cite{TZ01}, which provides a bound on the overlap of clusters.
\begin{claim}\label{claim:number-of-clusters}
With high probability, each vertex is contained in at most $4n^{1/k}\log n$ clusters.
\end{claim}

The following claim shows that (with high probability) the sets $A_i$ have favorable properties.
\begin{claim}\label{claim:hit-path}
With high probability the following holds for every $0\le i\le k-1$: (1) $|A_i|\le 4n^{1-i/k}\ln n$, and (2) For every $u,v\in V$ such that $h(u,v)>4n^{i/k}\ln n$, there exists a vertex of $A_i$ on the shortest path between $u$ and $v$.
\end{claim}
\begin{proof}
Fix $i$. The first assertion holds by a simple Chernoff bound, since every vertex is chosen to be in $A_i$ independently with probability $n^{-i/k}$, and the expected size of $A_i$ is $n^{1-i/k}$.
For the second assertion, let $u,v$ be such that $h(u,v)>4n^{i/k}\ln n$ (recall that $h(u,v)$ is the number of hops on the shortest path from $u$ to $v$ in $G$). The probability that none of the vertices on the $u$ to $v$ shortest path is included in $A_i$ is at most
\[
\left(1-n^{-i/k}\right)^{4n^{i/k}\ln n}\le n^{-4}~.
\]
Taking a union bound on the $k$ possible values of $i$ and ${n\choose 2}$ pairs completes the proof.
\end{proof}
From now on assume that all the events in the claims above hold, which yields the following corollary.
\begin{corollary}\label{cor:hit-path}
For any $0\le i< k-1$, $u\in A_i\setminus A_{i+1}$ and $v\in C(u)$, it holds that $h(u,v)\le 4n^{(i+1)/k}\ln n$.
\end{corollary}
\begin{proof}
If it were the case that $h(u,v)>4n^{(i+1)/k}\ln n$, then \claimref{claim:hit-path} would imply that there exists a vertex of $A_{i+1}$ on the shortest path from $v$ to $u$. In particular, $d_G(v,u)>d_G(v,A_{i+1})$, which contradicts \eqref{eq:cluster}.
\end{proof}

\subsection{Approximate Clusters and Pivots}

Since we do not know how to compute efficiently in a distributed manner the pivots and clusters, we settle for an approximate version, which is formally defined in this section. Fix the parameter $\epsilon=\frac{1}{48k^4}$. For each $v\in V$ and $0\le i\le k-1$, a point $\hat{z}\in A_i$ is called an {\em approximate $i$-pivot} of $v$ if
\begin{equation}\label{eq:pivot}
d_G(v,\hat{z})\le (1+\epsilon)d_G(v,A_i)~.
\end{equation}
Now we define for each $0\le i\le k-1$ and each vertex $u\in A_i\setminus A_{i+1}$, a set of vertices which we call an {\em approximate cluster}. The approximate cluster is a subset of the cluster $C(u)$, and it is allowed to exclude vertices of $C(u)$ which are "close" to the boundary. First define the vertices that are far from the boundary (with respect to $\epsilon$), as
\begin{equation}\label{eq:eps-clust}
C_\epsilon(u)=\{v\in V~:~ d_G(u,v)<\frac{d_G(v,A_{i+1})}{1+\epsilon}\}.
\end{equation}
The approximate cluster $\tilde{C}(u)$ will be a set that satisfies the following:
\begin{equation}\label{eq:app-cluster}
C_{6\epsilon}(u)\subseteq\tilde{C}(u)\subseteq C(u)~.
\end{equation}

Each approximate cluster $\tilde{C}(u)$ we compute, will be stored as a tree rooted at $u$, that is, each vertex $v\in \tilde{C}(u)$ will store a pointer to its parent in the tree. This tree (abusing notation, we call this tree $\tilde{C}(u)$ as well) has the property that distances to the root $u$ are approximately preserved, that is, for any $v\in\tilde{C}(u)$ we have that
\begin{equation}\label{eq:tree-preserve}
d_G(u,v)\le d_{\tilde{C}(u)}(u,v)\le(1+\epsilon)^4d_G(u,v)~.
\end{equation}

\begin{remark}\label{rem:siize}
Since $\tilde{C}(u)\subseteq C(u)$, \claimref{claim:number-of-clusters} implies that with high probability, each vertex is contained in at most $4n^{1/k}\log n$ approximate clusters.
\end{remark}

In the remainder of this section we devise an efficient distributed algorithm for computing the approximate pivots and the trees built from approximate clusters, and show the following.

\begin{theorem}\label{thm:app-pivot-cluster}
Let $G=(V,E)$ be a weighted graph with $n$ vertices and hop-diameter $D$, and let $k\ge 1$ be an integer. Set $\epsilon=1/(48k^4)$. Then there is a randomized distributed algorithm that w.h.p computes all approximate pivots and approximate clusters (with respect to $\epsilon$) within $(n^{1/2+1/k}+D)\cdot \min\{(\log n)^{O(k)},2^{\tilde{O}(\sqrt{\log n})}\}$ rounds.\footnote{For odd $k$ the number of rounds becomes $(n^{1/2+1/(2k)}+D)\cdot \min\{(\log n)^{O(k)},2^{\tilde{O}(\sqrt{\log n})}\}$.}
\end{theorem}

\paragraph{Computing Pivots.}
We first compute the pivots for $0\le i\le \lceil k/2\rceil$. For these values of $i$ we can compute the exact pivots. We conduct $4n^{i/k}\cdot \ln n$ iterations of Bellman-Ford rooted in the vertex set $A_i$. As a result, every $v\in V$ learns the exact value $\hat{d}_i(v)=d_G(v,A_i)$ and a pivot $\hat{z}_i(v)\in A_i$. Indeed, for any $v\in V$, if $u\in A_i$ is a vertex such that $d_G(v,u)=d_G(v,A_i)$, then \claimref{claim:hit-path} implies that $h(v,u)\le 4n^{i/k}\cdot \ln n$, so the exploration will detect this shortest path.
As every message consists of $O(1)$ words (every vertex sends to its neighbors the name of the vertex in $A_i$ and the current distance to it), the total number of rounds is $\sum_{i=0}^{\lceil k/2\rceil}O(n^{i/k}\cdot \ln n)\le\tilde{O}(n^{1/2+1/(2k)})$.

For $ \lceil k/2\rceil< i\le k-1$ we can only compute {\em approximate} pivots $\hat{z}_i(v)$ for each $v\in V$. For each such $i$, apply \theoremref{thm:SPT} with root set $A_i$ and the parameter $\epsilon$ (indeed by \claimref{claim:hit-path}, $|A_i|\le 4n^{1-(\lceil k/2\rceil+1)/k}\ln n\le 2\sqrt{n}\ln n$, and $\epsilon=\Omega(1/k^4)\ge\Omega(1/\log^4n)$). This will take $(n^{1/2+1/(2k)}+D)\cdot \min\{(\log n)^{O(k)},2^{\tilde{O}(\sqrt{\log n})}\}$ rounds. At the end, every vertex $v\in V$ will know its approximate pivot $\hat{z}_i(v)$, and the (approximate) distance $\hat{d}_i(v)$, as returned by the algorithm. By \eqref{eq:stretc}, $\hat{z}_i(v)$ satisfies the requirement from an approximate pivot (see \eqref{eq:pivot}).


\subsection{Building the Small Trees}
For $0\le i< \lceil k/2\rceil$, we can compute the trees $C(u)$ corresponding to the actual clusters. We need to find such a tree for every $u\in A_i\setminus A_{i+1}$, and it is done in the following manner. For each such $u$ in parallel, we initiate a bounded-depth Bellman-Ford exploration for $4n^{(i+1)/k}\ln n$ iterations. By bounded-depth we mean the following: each $v\in V$ that receives a message originated at $u$, and computes that its (current) distance to $u$ is $b_v(u)$, will join $C(u)$ and broadcast the message to its neighbors in $G$ iff
\begin{equation}\label{eq:w+b}
b_v(u) < d_G(v,A_{i+1})~.
\end{equation}
(Recall that for $i\le \lceil k/2\rceil$, each vertex stores the distance to the exact $i$-th pivot $\hat{d}_i(v)=d_G(v,A_i)$.) The vertex $v$ will also store the name of its parent in $C(u)$, the neighbor $p\in V$ that sent $v$ the message which last updated $b_v(u)$.

We now argue that if $v\in C(u)$, then $v$ will surely receive a message from $u$ and will have $b_v(u)=d_G(u,v)$.
Let $P$ be the shortest path in $G$ between $u$ and $v$. Note that every vertex $y$ on $P$ has $y\in C(u)$, because
\[
d_G(y,u)=d_G(v,u)-d_G(v,y)\stackrel{\eqref{eq:cluster}}{<}d_G(v,A_{i+1})-d_G(v,y)\le d_G(y,A_{i+1})~.
\]
It follows by a simple induction that every such $y$ will receive a message with the exact distance $b_y(u)=d_G(y,u)$ and thus will send it onwards, after at most $h(u,y)$ steps of the algorithm. In particular, distances to the root $u$ in $C(u)$ are preserved exactly. \corollaryref{cor:hit-path} asserts that for all $v\in C(u)$ we have that $h(u,v)\le 4n^{(i+1)/k}\ln n$. So there are enough Bellman-Ford iterations to reach all vertices of $C(u)$.

\paragraph{The middle level.}
When $k$ is odd, the level $i=(k-1)/2$  induces a relatively large running time $\tilde{O}(n^{1/2+3/(2k)})$ (see the upcoming paragraph on running-time analysis), if one uses the algorithm that was described above. To overcome this, we use a different method for this level. We apply \theoremref{thm:Nanongkai} on the set of sources $S=A_i\setminus A_{i+1}$, with $B=4n^{(i+1)/k}\cdot\ln n$ and $\epsilon$, each vertex $v\in V$ will get a distance estimate $b_v(u)$ for each $u\in S$. Indeed, if $v\in C(u)$ then by \corollaryref{cor:hit-path}, $h(u,v)\le B$, so that the distance estimate returned by the theorem is a $1+\epsilon$ approximation to $d_G(u,v)=d_G^{(B)}(u,v)$.

We say that $v$ joins the (approximate) cluster $\tilde{C}(u)$ of $u\in S$ if the following holds
\[
b_v(u)<d_G(v,A_{i+1}),
\]
(recall that $v$ knows the exact distance to its $i+1=(k+1)/2$-pivot). The parent $p$ of $v$ in the tree induced by $\tilde{C}(u)$ will be the parent given by \remarkref{rem:parents}. We show that this $p$ will join $\tilde{C}(u)$ as well. This holds because
\[
b_p(u)\stackrel{\eqref{eq:p-u}}{\le}b_v(u)-w(v,p)
<d_G(v,A_{i+1})-d_G(v,p)\le d_G(p,A_{i+1})~.
\]
Finally, we note that this is an approximate cluster; since $d_G(u,v)\le b_v(u)$ it follows that $\tilde{C}(u)\subseteq C(u)$, while if $v\in C_{\epsilon}(u)$ then
\[
b_v(u)\stackrel{\eqref{eq:duv}}{\le}(1+\epsilon)d_G(u,v)\stackrel{\eqref{eq:eps-clust}}{<}d_G(v,A_{i+1})~,
\]
so $\tilde{C}(u)\supseteq C_\epsilon(u)$, satisfying \eqref{eq:app-cluster}.
(We remark that the middle level is the only one in which one may use \theoremref{thm:Nanongkai}. In all other levels, either the number of sources $|A_i|\approx n^{1-i/k}$ or the required depth $B\approx n^{(i+1)/k}$ will be larger than $n^{1/2+1/k}$.)

\paragraph{Running time.} By \claimref{claim:number-of-clusters}, every vertex can belong to at most $\tilde{O}(n^{1/k})$ clusters. Hence, the congestion at every Bellman-Ford iteration is at most $\tilde{O}(n^{1/k})$. Thus the number of rounds required to implement each of the $4n^{(i+1)/k}\ln n$ iterations of Bellman-Ford is $\tilde{O}(n^{1/k})$. When $k$ is even, the total running time is $\sum_{i=0}^{ k/2-1}\tilde{O}(n^{(i+2)/k})=\tilde{O}(n^{1/2+1/k})$. When $k$ is odd, the middle level $(k-1)/2$ will take time $\tilde{O}(|S|+B+D)=\tilde{O}(n^{1/2+1/(2k)}+D)$, while the lower levels will take $\sum_{i=0}^{(k-3)/2}\tilde{O}(n^{(i+2)/k})=\tilde{O}(n^{1/2+1/(2k)})$. So for odd $k$, the total running time is $\tilde{O}(n^{1/2+1/(2k)}+D)$~.

\subsection{Building the Large Trees}

Building the trees $\tilde{C}(u)$ for $u\in A_i\setminus A_{i+1}$ when $i\ge \lceil k/2\rceil$ is more involved, since the number of iterations for the simple Bellman-Ford style approach grows like $\approx n^{(i+2)/k}$. We will use the fact that there are only few vertices in $A_i$, and divide the computation into two phases. In the first phase we compute virtual trees only on $\approx\sqrt{n}$ vertices, and in the second phase we extend the trees to the entire  graph. Before we turn to the two-phase construction, we describe the preprocessing stage, in which we build structures that are later used in both phases.

\subsubsection{Preprocessing}\label{sec:pre}

Let $V'=A_{\lceil k/2\rceil}$, and
set $B=4n/\E[|V'|]\cdot\ln n$. That is, for even $k$ we set $B=4n^{1/2}\cdot\ln n$, while for odd $k$, $B=4n^{1/2+1/(2k)}\cdot\ln n$.
Apply \theoremref{thm:Nanongkai} to $G$ with the set $V'$ and parameters $B$ and $\epsilon/2$. By \claimref{claim:hit-path} we may assume $|V'|\le 4n^{1/2}\ln n$, and since $1/\epsilon\le 48\log^4n$, the number of rounds required is w.h.p $\tilde{O}(n^{1/2+1/(2k)}+D)$. From now on assume that \eqref{eq:duv} indeed holds (with $\epsilon$ replaced by $\epsilon/2$). This happens w.h.p. Let $G'=(V',E',w')$ be a (virtual) graph on $G$, and for each $u,v\in V'$ with $d_{uv}<\infty$, set the weight of the edge connecting them to be $w'(u,v)=d_{uv}$ (where $d_{uv}$ is the value computed in \theoremref{thm:Nanongkai}).
Following \cite{N14}, it can be shown that for any $u,v\in V'$,
\begin{equation}\label{eq:g'}
d_G(u,v)\le d_{G'}(u,v)\le(1+\epsilon/2)d_G(u,v)~.
\end{equation}

Apply \theoremref{thm:hopset} on $G'$ with parameters $\epsilon/3$ and $\rho=\max\{1/k,\log\log n/\sqrt{\log n}\}$. We obtain a $(\beta,\epsilon/3)$-hopset $F$ with $\beta=\min\{2^{\tilde{O}(\sqrt{\log n})},(\log n)^{O(k)}\}$. The number of rounds required is $\tilde{O}(|V'|^{1+\rho}+D)\cdot \beta^2=(n^{1/2(1 +1/k)}+D)\cdot\min\{2^{\tilde{O}(\sqrt{\log n})},(\log n)^{O(k)}\}$.


Let $G''=(V',E'\cup F,w'')$ be the graph obtained from $G'$ by adding all the hopset edges. (Note that some edges may have their weight replaced. In the case of conflict, the weights $w''$ agree with the weights of the hopset $F$.) By \eqref{eq:hopset1} and \eqref{eq:g'} we have that $G''$ is indeed a virtual graph since
$d_{G''}(u,v)\ge d_{G'}(u,v)\ge d_G(u,v)$. On the other hand,
\begin{eqnarray*}
d_{G''}^{(\beta)}(u,v)&\le& (1+\epsilon/3)d_{G'}(u,v)\le (1+\epsilon/2)(1+\epsilon/3)d_G(u,v)\\
&\le& (1+\epsilon)d_G(u,v)~.
\end{eqnarray*}

We conclude that the graph $G''$ satisfies the following property: for every $u,v\in V'$,
\begin{equation}\label{eq:g''}
d_G(u,v)\le d_{G''}^{(\beta)}(u,v)\le(1+\epsilon)d_G(u,v)~.
\end{equation}

\subsubsection{Construction}

Fix $\lceil k/2\rceil\le i\le k-1$. We build the trees $\tilde{C}(u)$ for all $u\in A_i\setminus A_{i+1}$ in parallel, in two main phases.

\paragraph{Phase 1.}
For each such  $u$, conduct $\beta$ iterations of depth-bounded Bellman-Ford in the graph $G''$.\footnote{See \eqref{eq:rrr} below for the required condition on depth.} (Since this is a virtual graph, all the messages will be collected at the root of some BFS tree of $G$ via pipelined convergecast, and then broadcasted to the entire graph $G$ via pipelined broadcast. See \lemmaref{lem:pipe}.) If $v\in V'$ receives a message originated at $u$ with (current) distance to $u$ which is $b_v(u)$, it will join the approximate cluster of $u$ and forward the message to its neighbors in $G''$ iff
\begin{equation}\label{eq:rrr}
b_v(u)<\frac{\hat{d}_{i+1}(v)}{(1+\epsilon)^3}~.
\end{equation}
(Recall that $\hat{d}_{i+1}(v)$ is the approximate distance from $v$ to the its (approximate) level $i+1$ pivot.)
The vertex $v$ will also store its {\em virtual} parent, the neighbor $p\in V'$ that sent $v$ the message which last updated $b_v(u)$. For each $u\in A_i\setminus A_{i+1}$, we have a (virtual) tree $\tilde{C}'(u)$ on the vertices of $V'$ that received a message originated at $u$ and satisfy \eqref{eq:rrr}.

\paragraph{Phase 1.5.}
The purpose of this step is to guarantee that every vertex which was added to the (virtual) tree being built for some $u\in A_i\setminus A_{i+1}$, will have an appropriate parent in $G$ (through which it will route later on). The issue is that hopset edges are not equipped with parents in $G$, unlike the edges of $G'$, for which \remarkref{rem:parents} provides parents. We deal with this by using the path-reporting property of hopset edges -- each such edge is realized by a path in $G'$, so we ensure the vertices of this path join the tree as well, and set parents accordingly. We now describe this formally.

When the first phase ends after $\beta$ iterations, for every hopset edge $(x,y)\in F$ such that $x$ is the virtual parent of $y$ we do the following. Let $P$ be the path in $G'$ realizing this edge. Each $v\in V'(P) \setminus\{x\}$ that has $b_v(u)$ value (for some $u \in A_i \setminus A_{i+1}$) at least $b_x(u)+d_P(x,v)$, updates its distance estimate to be $b_v(u)=b_x(u)+d_P(x,v)$,  joins $\tilde{C}'(u)$ (if it hasn't already), and sets its virtual parent as $v'$, where $v'$ is the neighbor of $v$ on $P$ closer to $x$ (recall \propertyref{prop:hop}, which guarantees that $v$ knows the relevant information).

Finally, set the {\em real} parents: for each vertex $v\in \tilde{C}'(u)$ with a virtual parent $v'$, set $p(v)=p_{v'}(v)$ (see \remarkref{rem:parents} for the definition and computation of $p_{v'}(v)$). Recall that $(v,v')$ is a virtual edge (of the graph $G'$), while $(v,p(v))$ is a ``real'' edge from $G$.

\paragraph{Phase 2.} Here we extend each virtual tree $\tilde{C}'(u)$ to the vertices of $V$. For all $u\in A_i\setminus A_{i+1}$, every vertex $v\in \tilde{C}'(u)$ broadcasts to the entire graph its value $b_v(u)$ (and the name of $u$). A vertex $y\in V$ will add itself to $\tilde{C}(u)$ if
\begin{equation}\label{eq:y-t-u}
d_{yv}+b_v(u)<\frac{\hat{d}_{i+1}(y)}{1+\epsilon}~,
\end{equation}
where $d_{yv}$ is the value computed in \theoremref{thm:Nanongkai}.
Also, $y$ will set $p(y)=p_v(y)$ as its (real) parent in $\tilde{C}(u)$ for the $v$ minimizing $b_y(u)=d_{yv}+b_v(u)$ (breaking ties arbitrarily). We remark that the condition of \eqref{eq:y-t-u} is less stringent than that of \eqref{eq:rrr}. Thus vertices of $V'$ who did not join $\tilde{C}'(u)$, may now be included in $\tilde{C}(u)$. 

First we argue that for any $u\in A_i\setminus A_{i+1}$, the vertices $v\in V'$ added to $\tilde{C}'(u)$ in phase 1.5 with distance estimate $b_v(u)$ satisfy the following:
\begin{equation}\label{eq:rrro}
b_v(u)<\frac{\hat{d}_{i+1}(v)}{(1+\epsilon)^2}~.
\end{equation}
To see this, let $(x,y)\in F$ be the hop-set edge which triggered the addition of $v$ to $\tilde{C}'(u)$ at phase 1.5, and let $P$ be the path in $G'$ realizing this edge, then
\[
b_v(u)= d_P(x,v)+b_x(u) = d_P(x,y)-d_P(v,y)+b_x(u)=b_y(u)-d_P(v,y)~.
\]
It follows that
\[
b_v(u)= b_y(u)-d_P(v,y)\stackrel{\eqref{eq:rrr}}{<}\frac{\hat{d}_{i+1}(y)}{(1+\epsilon)^3}-d_G(v,y)\stackrel{\eqref{eq:stretc}}{\le}\frac{d_G(y,A_{i+1})-d_G(v,y)}{(1+\epsilon)^2}\le \frac{d_G(v,A_{i+1})}{(1+\epsilon)^2}\stackrel{\eqref{eq:stretc}}{\le}\frac{\hat{d}_{i+1}(v)}{(1+\epsilon)^2}~,
\]
which proves \eqref{eq:rrro}. The next lemma asserts that the values $b_v(u)$ approximate well the distances to the root $u$ of the virtual tree.
\begin{lemma}\label{lem:bv}
For any $u\in A_i\setminus A_{i+1}$ and $v\in\tilde{C}(u)$ with the corresponding value $b_v(u)$, we have that
\begin{equation}\label{eq:ttre}
d_G(u,v)\le b_v(u)\le(1+\epsilon)^4d_G(u,v)~.
\end{equation}
\end{lemma}
\begin{proof}
First we prove for $v\in\tilde{C}'(u)$ added at phase 1. Note that the left hand side of \eqref{eq:ttre} can be verified by induction on the iteration in which $b_v(u)$ was last updated. The base case $u=v$ clearly holds, assume it holds for $v'$ (the virtual parent of $v$). Recall that $w''$ is the weight function in $G''$. We have
\[
b_v(u)=w''(v,v')+b_{v'}(u)\ge d_{G''}(v,v')+d_G(u,v')\stackrel{\eqref{eq:g''}}{\ge} d_G(u,v)~.
\]

We now turn to the right hand side of \eqref{eq:ttre}. Seeking contradiction, assume
\begin{equation}
\label{eq:bvu}
b_v(u)>(1+\epsilon)^4d_G(u,v)~.
\end{equation}
 Let $P$ be the shortest $\beta$-hops path in $G''$ from $u$ to $v$, and we will show (by induction) that every vertex $z$ on $P$, which lies $h$ hops from $u$, must join $\tilde{C}'(u)$ with value $b_z(u)\le d_P(u,z)$ by the  iteration $h$  of the Bellman-Ford exploration of phase 1. The base case for $z=u$ clearly holds. Fix any other $z\in P$ with $h$ hops to $u$ on $P$, and assume it holds for $p$, the neighbor of $z$ on $P$ (the one closer to $u$), so we have that $b_p(u)\le d_P(u,p)$ by iteration $h-1$. At iteration $h$, $p$ will broadcast its value $b_p(u)$, and thus $z$ could have updated its value to be $b_p(u)+w''(p,z)$. In particular,
\begin{equation}\label{eq:bz}
b_z(u)\le b_p(u)+w''(p,z)\le d_P(u,p)+w''(p,z)=d_P(u,z).
\end{equation}
We now argue $b_z(u)$ satisfies \eqref{eq:rrr}, which would cause $z$ to join $\tilde{C}'(u)$,
\begin{eqnarray}\nonumber
b_z(u)&\stackrel{\eqref{eq:bz}}{\le}&d_P(u,z)\\\nonumber
&=&d_P(u,v)-d_P(v,z)\\\label{eq:rerere}
&\le& d_{G''}^{(\beta)}(u,v)-d_G(v,z)\\\nonumber
&\stackrel{\eqref{eq:g''}}{\le}&(1+\epsilon)d_G(u,v)-d_G(v,z)\\\label{eq:rerer}
&\stackrel{(\ref{eq:bvu})}{\le}&\frac{b_v(u)}{(1+\epsilon)^2}-d_G(v,z)\\\label{eq:terer}
&\stackrel{\eqref{eq:rrr}}{<}&\frac{\hat{d}_{i+1}(v)}{(1+\epsilon)^4}-d_G(v,z)\\\nonumber
&\stackrel{\eqref{eq:stretc}}{\le}&\frac{d_G(v,A_{i+1})-d_G(v,z)}{(1+\epsilon)^3}\\\nonumber
&\le&\frac{d_G(z,A_{i+1})}{(1+\epsilon)^3}\\\nonumber
&\stackrel{\eqref{eq:stretc}}{\le}&\frac{\hat{d}_{i+1}(z)}{(1+\epsilon)^3}~,
\end{eqnarray}
where \eqref{eq:rerere} uses that $P$ is the shortest $\beta$-hops path in $G''$, and \eqref{eq:rerer} uses the contradiction assumption (\ref{eq:bvu}) (note that it was used with the term $(1+\epsilon)^3$ rather than $(1+\epsilon)^4$).
Hence $z$ joins $\tilde{C}'(u)$, and so $b_v(u) \le d_P(u,v)$. Hence
\[
b_v(u)\le d_P(u,v)= d_{G''}^{(\beta)}(u,v)\stackrel{\eqref{eq:g''}}{\le}(1+\epsilon)d_G(u,v),
\]
which contradicts our assumption that \eqref{eq:ttre} does not hold.

 We now turn to vertices $v\in\tilde{C}'(u)$ who joined in phase 1.5. The left hand side holds since if $(x,y)\in F$ is the hop-set edge that triggered the addition of $v$, and $P'$ is the path in $G'$ realizing this edge, we have that $b_v(u)=d_{P'}(v,x)+b_x(u)\ge d_G(v,x)+d_G(x,u)\ge d_G(v,u)$. For the right hand side, note that we only used the fact that $v$ joined in phase 1 at \eqref{eq:terer}, so we can repeat the argument, replacing the use of \eqref{eq:rrr} by \eqref{eq:rrro}. We indeed lose a factor of $1+\epsilon$, but the inequality is still valid, yielding the same contradiction.

Finally, we turn to $v\in \tilde{C}(u)$ joining at phase 2. Note that for each such $v$, there exists some $x\in V'$ for which $v$ sets its value to be $b_v(u)=d_{vx}+b_x(u)\ge d_G(v,x)+d_G(x,u)\ge d_G(v,u)$, which proves the left hand side of \eqref{eq:ttre}. For the right hand side, consider first the case that $h(v,u)\le B$. Since $v$ could update $b_v(u)$ directly from the broadcast of $u$ itself, we have
\[
b_v(u)\le 0+d_{vu}\stackrel{\eqref{eq:duv}}{\le}(1+\epsilon)d_G^{(B)}(v,u)=(1+\epsilon)d_G(v,u)~.
\]
The other case is when $h(v,u)>B$, but then \claimref{claim:hit-path} (with $i=\lceil k/2\rceil$) suggests that there exists $x\in V'$ on the shortest path in $G$ from $v$ to $u$, with $h(v,x)\le B$. In particular, $d_G^{(B)}(x,v)=d_G(x,v)$. Again seeking contradiction, assume \eqref{eq:ttre} does not hold for $v$. Let $P$ be the shortest (at most) $\beta$-hops path from $u$ to $x$ in $G''$. We claim that every $z\in P$ must have joined $\tilde{C}'(u)$ at phase 1. To see this by induction, fix $z\in P$ with $h$ hops from $u$ on $P$, and assume $p$ (the neighbor of $z$ closer to $u$) did join by the $h-1$ iteration of Bellman-Ford, with $b_p(u)\le d_P(u,p)$. When $p$ broadcasts $b_p(u)$ at step $h$, then indeed $b_z(u)\le b_p(u)+w''(p,z)=d_P(u,z)$. Now,
\begin{eqnarray}
\nonumber
b_z(u)&\le& d_P(u,z)\\
\nonumber
&\le&d_{G''}^{(\beta)}(u,x)-d_P(z,x)\\
\label{eq:third}
&\stackrel{\eqref{eq:g''}}{\le}&(1+\epsilon)d_G(u,x)-d_G(z,x)\\
\nonumber
&=&(1+\epsilon)[d_G(u,v)-d_G(x,v)]-d_G(z,x)\\
\nonumber
&\le&\frac{b_v(u)}{(1+\epsilon)^3}-d_G(x,v)-d_G(z,x)\\
\nonumber
&\stackrel{\eqref{eq:y-t-u}}{<}&\frac{\hat{d}_{i+1}(v)}{(1+\epsilon)^4}-d_G(x,v)-d_G(z,x)\\
\nonumber
&\stackrel{\eqref{eq:pivot}}{\le}&\frac{d_G(v,A_{i+1})-d_G(x,v)-d_G(z,x)}{(1+\epsilon)^3}\\
\nonumber
&\le&\frac{d_G(z,A_{i+1})}{(1+\epsilon)^3}\\
\nonumber
&\le&\frac{\hat{d}_{i+1}(z)}{(1+\epsilon)^3}~.
\end{eqnarray}
((\ref{eq:third})  is because $x$ lies on the shortest $u-v$ path in $G$.)

This implies $b_z(u)$ satisfies \eqref{eq:rrr} and thus $z$ indeed joins $\tilde{C}'(u)$ by iteration $h$ of phase 1. In particular, $x$ joins by the end of phase 1, and broadcasts $b_x(u)$ at phase 2. Then we have that
\[
b_v(u)\le b_x(u)+d_{xv}\stackrel{\eqref{eq:ttre}}{\le}(1+\epsilon)^4d_G(u,x)+(1+\epsilon)d_G^B(x,v)\le(1+\epsilon)^4d_G(u,v)~,
\]
(Recall that $d_{xv}$ is the value computed by the algorithm of Theorem \ref{thm:Nanongkai}.)
This yields a contradiction  to (\ref{eq:bvu}) and concludes the proof.
\end{proof}

The following lemma shows that the sets $\tilde{C}(u)$ satisfy the requirement from approximate clusters. The proof is similar to that of \lemmaref{lem:bv}, though it uses the definition of $C_\epsilon(u)$, rather than the (contradiction) assumption that $b_v(u)$ is large.
\begin{lemma}\label{lem:app-cluster}
For any $u\in A_i\setminus A_{i+1}$, the set $\tilde{C}(u)$ satisfies \eqref{eq:app-cluster}.
\end{lemma}
\begin{proof}

For the right hand side of \eqref{eq:app-cluster}, note that if $v\in\tilde{C}'(u)$, then
\[
d_G(u,v)\stackrel{\eqref{eq:ttre}}{\le}b_v(u)\stackrel{\eqref{eq:rrr}\wedge\eqref{eq:rrro}}{<}\frac{\hat{d}_{i+1}(v)}{(1+\epsilon)^2} \stackrel{\eqref{eq:stretc}}{\le}d_G(v,A_{i+1})~,
\]
so $v\in C(u)$ as well.
 For the left hand side of \eqref{eq:app-cluster} (at this point we only show that $\tilde{C}'(u)\supseteq C_{6\epsilon}(u)\cap V'$), consider $v\in C_{6\epsilon}(u)\cap V'$, and let $P$ be the (at most) $\beta$-hops shortest path from $v$ to $u$ in $G''$. It suffices to show that every vertex $y$ along this path which is $h$ hops from $u$, will join $\tilde{C}'(u)$ and have $b_y(u)\le d_P(y,u)$ by the  iteration  $h$ of Bellman-Ford in phase 1. Assume (by induction) that $p$, the predecessor of $y$ on $P$, joins $\tilde{C}'(u)$ and satisfies $b_p(u)\le d_P(p,u)$ by iteration $h-1$. Thus, $p$ sends at iteration $h$ the value $b_p(u)$. Since $b_y(u)\le w''(y,p)+b_p(u)\le w''(y,p)+d_P(u,p)=d_P(u,y)$, it remains to show that this value of $b_y(u)$ satisfies \eqref{eq:rrr}, and thus $y$ joins $\tilde{C}'(u)$. To this end,

\begin{eqnarray*}
b_y(u)&\le&d_P(u,y)\\
&\le&d_{G''}^{(\beta)}(u,v)-d_P(y,v)\\
&\stackrel{\eqref{eq:g''}}{\le}&(1+\epsilon)d_G(u,v)-d_G(y,v)\\
&\le&\frac{(1+\epsilon)d_G(v,A_{i+1})}{1+6\epsilon}-d_G(y,v)\\
&<&\frac{d_G(v,A_{i+1})-d_G(y,v)}{(1+\epsilon)^3}\\
&\le&\frac{d_G(y,A_{i+1})}{(1+\epsilon)^3}\\
&\stackrel{\eqref{eq:stretc}}{\le}&\frac{\hat{d}_{i+1}(y)}{(1+\epsilon)^3}~.
\end{eqnarray*}
where the fourth inequality uses that $v\in C_{6\epsilon}(u)$ (recall \eqref{eq:eps-clust}).
 This implies $v$ will join $\tilde{C}'(u)$ in phase 1.

We now prove that \eqref{eq:app-cluster} holds for $\tilde{C}(u)$. For the right hand side, let $y\in \tilde{C}(u)\setminus \tilde{C}'(u)$, then there exists $v\in V'$ for which $y$ satisfies \eqref{eq:y-t-u}. So we obtain
\[
d_G(y,u)\le d_G(y,v)+d_G(v,u)\stackrel{\eqref{eq:duv}\wedge\eqref{eq:ttre}}{\le}d_{yv}+b_v(u)\stackrel{\eqref{eq:y-t-u}}{<}\frac{\hat{d}_{i+1}(y)}{1+\epsilon}\stackrel{\eqref{eq:stretc}}{\le} d_G(y,A_{i+1})~.
\]
This implies that $y\in C(u)$. For the left hand side of \eqref{eq:app-cluster}, assume $y\in C_{6\epsilon}(u)$. Consider first the case that $h(u,y)\le B$. Then when $u$ broadcasts $b_u(u)=0$ at phase 2, $y$ will add itself to $\tilde{C}(u)$ because
\begin{equation}\label{eq:upb-y}
d_{yu}+0\stackrel{\eqref{eq:duv}}{\le}(1+\epsilon)d_G^{(B)}(y,u)=(1+\epsilon)d_G(y,u) \stackrel{\eqref{eq:eps-clust}}{\le} \frac{1+\epsilon}{1+6\epsilon}\cdot d_G(y,A_{i+1})\stackrel{\eqref{eq:stretc}}{<}\frac{\hat{d}_{i+1}(y)}{1+\epsilon}~.
\end{equation}
The other case is that $h(y,u)>B$. Then by \claimref{claim:hit-path} there is a vertex $v\in V'$ on the shortest path from $y$ to $u$ so that $h(y,v)\le B$. We now argue that $v\in \tilde{C}'(u)$, by a similar (though slightly more involved) argument as above. To see this, consider the shortest path $P$ with (at most) $\beta$-hops in $G''$ from $u$ to $v$, and we claim that each vertex $z$ on this path with $h$ hops from $u$, will join $\tilde{C}'(u)$ with $b_z(u)\le d_P(u,z)$ by iteration $h$ of the Bellman-Ford of phase 1. Again by induction, at step $h$ the vertex $z$ heard $b_p(u)\le d_P(u,p)$ from its predecessor $p$ on $P$. Then indeed $b_z(u)\le b_p(u)+w''(p,z)\le d_P(u,z)$. Now we show that $z$ joins $\tilde{C}'(u)$.
\begin{eqnarray}\nonumber
b_z(u)&\le& d_P(u,z)\\\nonumber
&=&d_{G''}^{(\beta)}(u,v)-d_P(z,v)\\\nonumber
&\stackrel{\eqref{eq:g''}}{\le}&(1+\epsilon)d_G(u,v)-d_G(z,v)\\\label{eq:fdfd}
&=&(1+\epsilon)[d_G(u,y)-d_G(y,v)]-d_G(z,v)\\\label{eq:fdfdf}
&\le&\frac{(1+\epsilon)d_G(y,A_{i+1})}{1+6\epsilon}-d_G(y,v)-d_G(z,v)\\\nonumber
&\le&\frac{d_G(y,A_{i+1})-d_G(y,v)-d_G(z,v)}{(1+\epsilon)^3}\\\nonumber
&\le&\frac{d_G(z,A_{i+1})}{(1+\epsilon)^3}\\\nonumber
&\stackrel{\eqref{eq:stretc}}{\le}&\frac{\hat{d}_{i+1}(z)}{(1+\epsilon)^3}~,
\end{eqnarray}
where \eqref{eq:fdfd} uses that $v$ is on the shortest path in $G$ from $u$ to $y$, and \eqref{eq:fdfdf} uses that $y\in C_{6\epsilon}(u)$.
In particular, we have shown $v\in \tilde{C}'(u)$ by the end of phase 1. It follows that $v$ will broadcast the value $b_v(u)\le d_{G''}^{(\beta)}(u,v)$ in the second phase. Since $h(y,v)\le B$,
\begin{eqnarray*}\label{eq:eq1}
b_y(u)&\le&d_{yv}+b_v(u)\\
&\stackrel{\eqref{eq:duv}}{\le}&(1+\epsilon)d_G^B(y,v) + d_{G''}^{(\beta)}(u,v)\\
&\stackrel{\eqref{eq:g''}}{\le}&(1+\epsilon)[d_G(y,v) + d_G(u,v)]\\
&=&(1+\epsilon)d_G(y,u)\\
&\stackrel{\eqref{eq:eps-clust}}{\le}&\frac{1+\epsilon}{1+6\epsilon}\cdot d_G(y,A_{i+1})\\
&<&\frac{\hat{d}_{i+1}(y)}{1+\epsilon}~.
\end{eqnarray*}

So $y$ will be added to $\tilde{C}(u)$. This concludes the proof of the lemma.

\end{proof}

Our next goal to to argue that the parent setting ensures that root-vertex distances in each cluster tree satisfy \eqref{eq:tree-preserve}, i.e., are approximated up to a factor $(1+\eps)^4$. It suffices to prove the following claim.

\begin{claim}\label{claim:par}
For any $u\in A_i\setminus A_{i+1}$, and any $v\in\tilde{C}(u)$, if $p=p(v)$ is the (real) parent of $v$ with corresponding value $b_p(u)$, then $p\in \tilde{C}(u)$ and
\begin{equation}\label{eq:show-parent}
b_v(u)\ge w(v,p)+b_p(u)~.
\end{equation}
\end{claim}
Once this claim is established, we get by induction on the depth of the tree that $d_{\tilde{C}(u)}(u,v)\le b_v(u)$. The base case when $u=v$ clearly holds, assume for $p=p(v)$ that $d_{\tilde{C}(u)}(u,p)\le b_p(u)$, and now
\[
d_{\tilde{C}(u)}(u,v)=w(v,p)+d_{\tilde{C}(u)}(u,p)\le w(v,p)+b_p(u)\stackrel{\eqref{eq:show-parent}}{\le}b_v(u)~.
\]
Combining this with \lemmaref{lem:bv} establishes \eqref{eq:tree-preserve}.

\begin{proof}[Proof of \claimref{claim:par}]

Consider first the case that $v\in \tilde{C}'(u)$, and there are two sub-cases to consider. In the first sub-case, $v$ updated $b_v(u)$ in phase 1 from some $x\in \tilde{C}'(u)$, who sent $b_x(u)$ over the (virtual) edge $(x,v)\in E'$ (which is not a hop-set edge). Then by the definition of $G'$,  $b_v(u)=w'(x,v)+b_x(u)=d_{xv}+b_x(u)$, the virtual parent of $v$ is set to $x$, and the real parent is thus $p=p_x(v)$. Since $p$ receives a message from $x$ in the second phase, it sets $b_p(u)$ to  at most $d_{px}+b_x(u)$. It follows that
\begin{equation}\label{eq:oop}
b_p(u)\le d_{px}+b_x(u)\stackrel{\eqref{eq:p-u}}{\le}d_{vx}-w(v,p)+b_x(u)=b_v(u)-w(v,p)~,
\end{equation}
which satisfies \eqref{eq:show-parent}. But we must also argue that $p$ indeed joins the tree $\tilde{C}(u)$. Here we use the relaxed condition of \eqref{eq:y-t-u} (compared to \eqref{eq:rrr}), and obtain that
\begin{eqnarray}\label{eq:donn}
b_p(u)&\stackrel{\eqref{eq:oop}}{\le}& b_v(u)-w(v,p)\\\label{eq:donn1}
&\stackrel{\eqref{eq:rrr}}{<}&\frac{\hat{d}_{i+1}(v)}{(1+\epsilon)^3}-d_G(v,p)\\\label{eq:donn2}
&\stackrel{\eqref{eq:pivot}}{\le}& \frac{d_G(v,A_{i+1})-d_G(v,p)}{1+\epsilon}\\\nonumber
&\le& \frac{d_G(p,A_{i+1})}{1+\epsilon}\\\nonumber
&\le& \frac{\hat{d}_{i+1}(p)}{1+\epsilon}~,
\end{eqnarray}
which satisfies \eqref{eq:y-t-u}.

The second sub-case is that $v$ updated $b_v(u)$ in phase 1 or 1.5 due to some hop-set edge $(x,y)\in F$, so that $v$ lies on the path $P$ in $G'$ realizing this edge (it could be that $y=v$, if it happened in phase 1). We set $b_v(u)= b_x(u)+d_P(x,v)$, and the virtual parent of $v$ is $v'\in V'$, its neighbor on $P$ which is closer to $x$.  Recall that in $G'$, the weight $w'(v,v')=d_{vv'}$, so that
\begin{equation}\label{eq:ppatgh}
d_P(x,v)=d_P(x,v')+d_{vv'}~.
\end{equation}
The real parent of $v$ is set as $p=p_{v'}(v)$. Since $v'$ broadcasts in phase 2 its estimate $b_{v'}(u)\le b_x(u)+d_P(x,v')$, it follows that
\begin{eqnarray*}
b_p(u)&\le& d_{pv'} +b_{v'}(u)\\
&\stackrel{\eqref{eq:p-u}}{\le}&(d_{vv'}-w(v,p))+(b_x(u)+d_P(x,v'))\\
&\stackrel{\eqref{eq:ppatgh}}{=}& b_x(u)+d_P(v,x)-w(v,p)\\
&=& b_v(u)-w(v,p)~,
\end{eqnarray*}
as required in \eqref{eq:show-parent}. Again, to see that $p\in \tilde{C}(u)$, we repeat the calculation of \eqref{eq:donn} with one change: In \eqref{eq:donn1}, replace the use of \eqref{eq:rrr} by \eqref{eq:rrro}, which will have the factor of $(1+\epsilon)^3$ replaced by $(1+\epsilon)^2$, but this suffices to satisfy \eqref{eq:donn2}.

We turn to the case that $v\in \tilde{C}(u) \setminus \tilde{C}'(u)$. Let $x\in \tilde{C}'(u)$ be the vertex which broadcasts in phase 2 a value $b_x(u)$ minimizing $b_v(u)=d_{vx}+b_x(u)$. The parent of $v$ is thus set to be $p=p_x(v)$, and now
\[
b_p(u)\le d_{px}+b_x(u)\stackrel{\eqref{eq:p-u}}{\le}d_{vx}-w(v,p)+b_x(u)=b_v(u)-w(v,p)~,
\]
The proof that $p\in\tilde{C}(u)$ is again similar to \eqref{eq:donn}.

\end{proof}


\paragraph{Running Time.} We noted that the number of rounds required for the preprocessing is $\tilde{O}(n^{1/2+1/(2k)}+D)\cdot\min\{2^{\tilde{O}(\sqrt{\log n})},(\log n)^{O(k)}\}$. Since by \eqref{eq:app-cluster} we have $\tilde{C}'(u)\subseteq C(u)$, then \remarkref{rem:siize} suggests that $v\in V'$ sends at most $\tilde{O}(n^{1/k})$ distance estimates $b_v(\cdot)$. As $|V'|\le \tilde{O}(n^{1/2})$, by \lemmaref{lem:pipe}, implementing a single Bellman-Ford iteration will take $\tilde{O}(n^{1/2+1/k}+D)$ rounds. As there are $\beta$ iterations in phase 1 (and a single one in phases 1.5 and 2), the total number of rounds is $\tilde{O}(n^{1/2+1/k}+D)\cdot\min\{2^{\tilde{O}(\sqrt{\log n})},(\log n)^{O(k)}\}$. (For odd $k$, both $|V'|\cdot n^{1/k},B\le \tilde{O}(n^{1/2+1/(2k)})$, so we get $\tilde{O}(n^{1/2+1/(2k)}+D)\cdot\min\{2^{\tilde{O}(\sqrt{\log n})},(\log n)^{O(k)}\}$ rounds.)


\section{Routing Based on Approximate Clusters}\label{sec:ana}

In this section we show that approximate pivots and approximate clusters suffice for a compact routing scheme, and prove our main result.

\begin{theorem}\label{thm:routing}
Let $G=(V,E)$ be a weighted graph with $n$ vertices and hop-diameter $D$, and let $k\ge 1$ be a parameter. Then there exists a routing scheme with stretch at most $4k-5+o(1)$, labels of size $O(k\log^2n)$ and routing tables of size $O(n^{1/k}\log^2n)$, that can be computed in a distributed manner within $(n^{1/2+1/k}+D)\cdot \min\{(\log n)^{O(k)},2^{\tilde{O}(\sqrt{\log n})}\}$ rounds, and for odd $k$ only $(n^{1/2+1/(2k)}+D)\cdot \min\{(\log n)^{O(k)},2^{\tilde{O}(\sqrt{\log n})}\}$ rounds.
\end{theorem}

\paragraph{Construction.}
Apply \theoremref{thm:app-pivot-cluster} on $G$ to obtain approximate pivots and approximate clusters for all vertices. For each $0\le i\le k-1$ and each $u\in A_i\setminus A_{i+1}$, construct the routing scheme for trees given by \theoremref{thm:tree-route} on $\tilde{C}(u)$.
(We postpone the proof of Theorem \ref{thm:tree-route}, i.e., the description of the algorithm that constructs routing tables and labels for each tree, to Section \ref{sec:tree}.) Specifically,  in each tree, every vertex stores a table of size $O(\log n)$ and has a label of size $O(\log^2n)$. The routing table of each $v\in V$ consists of all the tree-routing tables, for every $u\in V$ such that $v\in \tilde{C}(u)$. The label of $v$ consists of the tree-labels for the (at most) $k$ trees $\tilde{C}(\hat{z}_0(v)),\dots,\tilde{C}(\hat{z}_{k-1}(v))$, where $\hat{z}_i(v)$ is the approximate $i$-pivot of $v$ (note that it could be that $v$ does not belong to some of these trees, the label of $v$ will mark these as missing).
By \remarkref{rem:siize} there are at most $O(n^{1/k}\log n)$ trees containing $v$, and as each tree-table is of size $O(\log n)$, the routing table size is as promised. Since each tree-label is of size $O(\log^2n)$, the label size also obeys the given bound.

\paragraph{Finding a Tree.}
Assume we would like to route from vertex $u$ to vertex $v$. The routing protocol will find a vertex $w=\hat{z}_i(v)$ for some $0\le i\le k-1$, such that the stretch of the (unique) path from $u$ to $v$ in the tree $\tilde{C}(w)$ is at most $4k-5+o(1)$. 
The algorithm to find such a vertex appears in \algref{alg:find-tree}. 

\begin{algorithm}
\caption{$\texttt{Find-tree}(u,v)$}\label{alg:find-tree}
\begin{algorithmic}[1]
\STATE $i\leftarrow 0$;
\WHILE {$|\{u,v\}\cap\tilde{C}(\hat{z}_i(v))|<2$}
\STATE $i\leftarrow i+1$;
\ENDWHILE
\RETURN $\hat{z}_i(v)$;
\end{algorithmic}
\end{algorithm}


We note that our algorithm differs slightly from that of \cite{TZ01-spaa}, since it could be the case that $v$ does not belong to the cluster centered at the pivot of $v$ at level $i$. For this reason we keep searching until we find a cluster containing both $u,v$.

First we claim that the algorithm is correct. Note that the definition of approximate cluster \eqref{eq:app-cluster} implies that $\tilde{C}(x)=V$ for every $x\in A_{k-1}$ (this holds since the distance to $A_k$ is defined as $\infty$). Therefore when $i=k-1$ it must be that both $u,v\in\tilde{C}(\hat{z}_{k-1}(v))$, and the algorithm indeed halts. The tree $\tilde{C}(w)$ contains both $u,v$ (where $w=\hat{z}_i(v)$ is the vertex returned by the algorithm), by definition. Finally, the information from the label of $v$ indicates which of these trees contain it, and the routing table of $u$ also lists the names of all trees containing it. So we can run the algorithm from $u$ knowing the label of $v$.

Once $u$ computes the root $w$, it appends $w$ to the message header along with the label of $v$. From this point on the header does not change, and we route in the tree $\tilde{C}(w)$. Since this routing is exact, it remains to bound the stretch incurred by using the tree.

\paragraph{Bounding Stretch.}
We distinguish between two types of iterations $i$ that the algorithm did not stop at. Let $I_u=\{0\le i\le k-1~:~u\notin \tilde{C}(\hat{z}_i(v))\}$ be the iterations in which $\{u,v\}\cap\tilde{C}(\hat{z}_i(v))$ is empty or contains just $v$, and let $I_v=\{0\le i\le k-1~:~\{u,v\}\cap\tilde{C}(\hat{z}_i(v))=\{u\}\}$ be the remaining iterations in which the algorithm did not halt. For any $i\in I_u$, by \eqref{eq:app-cluster} it holds that $C_{6\epsilon}(\hat{z}_i(v))\subseteq\tilde{C}(\hat{z}_i(v))$. Hence, we have $u\notin C_{6\epsilon}(\hat{z}_i(v))$, which suggests that
\begin{eqnarray}\label{eq:stretch-u}
d_G(u,\hat{z}_{i+1}(u))&\stackrel{\eqref{eq:pivot}}{\le}&(1+\epsilon)d_G(u,A_{i+1})\nonumber\\ &\stackrel{\eqref{eq:eps-clust}}{\le}&(1+\epsilon)(1+6\epsilon)d_G(u,\hat{z}_i(v))\nonumber\\
&\le&(1+8\epsilon)d_G(u,\hat{z}_i(v))~.
\end{eqnarray}
Similarly for $i\in I_v$,
\begin{eqnarray}\label{eq:stretch-v}
d_G(v,\hat{z}_{i+1}(v))&\le&(1+\epsilon)d_G(v,A_{i+1})\nonumber\\
&\le&(1+\epsilon)(1+6\epsilon)d_G(v,\hat{z}_i(v))\nonumber\\
&\le&(1+8\epsilon)d_G(v,\hat{z}_i(v))~.
\end{eqnarray}

Define the following values $y_0=d_G(u,v)$, $x_0=0$, and for $0<i\le k-1$ define recursively $y_i=(1+10\epsilon)[y_0+x_{i-1}]$, and $x_i=(1+\epsilon)[y_0+y_i]$. Assume that the algorithm halted at iteration $i'$. Then for each $0\le i\le i'$ we claim that
\begin{equation}\label{eq:xi}
d_G(v,\hat{z}_i(v))\le x_i~.
\end{equation}
We verify the validity of \eqref{eq:xi} by induction. The base case trivially holds since $\hat{z}_0(v)=v$ and $x_0=0$. Fix $0<i\le i'$. The algorithm did not halt at iteration $i-1$. If it is the case that $i-1\in I_u$, then we have that
\begin{eqnarray}\label{eq:yi}
d_G(u,\hat{z}_i(u))&\stackrel{\eqref{eq:stretch-u}}{\le}&(1+8\epsilon)d_G(u,\hat{z}_{i-1}(v))\\
&\le&\nonumber(1+8\epsilon)[d_G(u,v)+d_G(v,\hat{z}_{i-1}(v))]\\
&\stackrel{\eqref{eq:xi}}{\le}&\nonumber(1+8\epsilon)[y_0+x_{i-1}]\\
&\le&\nonumber y_i~.
\end{eqnarray}
The other case is that $i-1\in I_v$. Since $\hat{z}_i(u)\in A_i$ we obtain
\begin{eqnarray}\label{eq:yiv}
d_G(u,\hat{z}_i(u))&\stackrel{\eqref{eq:pivot}}{\le}&(1+\epsilon)d_G(u,A_i)\\\nonumber
&\le&(1+\epsilon)d_G(u,\hat{z}_i(v))\\\nonumber
&\le&(1+\epsilon)[d_G(u,v)+d_G(v,\hat{z}_i(v))]\\\nonumber
&\stackrel{\eqref{eq:stretch-v}}{\le}&(1+\epsilon)[d_G(u,v)+(1+8\epsilon)d_G(v,\hat{z}_{i-1}(v))]\\\nonumber
&\le&(1+10\epsilon)[y_0+x_{i-1}]\\\nonumber
&=&y_i
\end{eqnarray}

We conclude that in both cases,
\begin{eqnarray}\label{eq:piv}
d_G(v,\hat{z}_i(v))&\le& (1+\epsilon)d_G(v,A_i)\\
&\le&\nonumber(1+\epsilon)d_G(v,\hat{z}_i(u))\\
&\le&\nonumber (1+\epsilon)[d_G(u,v)+d_G(u,\hat{z}_i(u))]\\
&\stackrel{\eqref{eq:yi}\wedge\eqref{eq:yiv}}{\le}&\nonumber(1+\epsilon)[y_0+y_i]\\
&=&\nonumber x_i~.
\end{eqnarray}

We now have a recurrence $x_i=(1+\epsilon)(2+10\epsilon)y_0 + (1+\epsilon)(1+10\epsilon)x_{i-1}$. Solving it, yields
\[
x_i = (1+\epsilon)(2+10\epsilon)y_0\sum_{j=0}^{i-1}[(1+\epsilon)(1+10\epsilon)]^j~.
\]
We use the fact that for any real $x\ge 0$ and positive integer $r$ such that $xr\le 1/2$, the following holds $(1+x)^r\le 1+2xr$. Now we may bound $x_i$ by
\begin{eqnarray}\label{eq:boundxi}
x_i&\le&(2+13\epsilon)y_0\sum_{j=0}^{i-1}(1+12\epsilon)^j\\\nonumber
&\le&(2+13\epsilon)y_0\sum_{j=0}^{i-1}(1+24\epsilon j)\\\nonumber
&\le&(2+13\epsilon)y_0(i+12\epsilon i^2)\\\nonumber
&\le&(2+13\epsilon)y_0(i+1/(4k^2))~,
\end{eqnarray}
where in the last inequality we use that $\epsilon=\frac{1}{48k^4}\le\frac{1}{48k^2i^2}$.
Finally, using that $i'\le k-1$ and that $w=\hat{z}_{i'}(v)$, the stretch is given by
\begin{eqnarray*}
\lefteqn{d_{\tilde{C}(w)}(u,w)+d_{\tilde{C}(w)}(w,v)}\\&\stackrel{\eqref{eq:tree-preserve}}{\le}& (1+\epsilon)^4[d_G(u,w)+d_G(v,w)]\\
&\stackrel{\eqref{eq:xi}}{\le}&(1+5\epsilon)[d_G(u,v)+2x_{i'}]\\
&\stackrel{\eqref{eq:boundxi}}{\le}& (1+5\epsilon)[1+(4+26\epsilon)(k-1+1/(4k^2))]\cdot d_G(u,v)\\
&\le& (4k-3+o(1))\cdot d_G(u,v)~.
\end{eqnarray*}

In order to improve the stretch to the promised $4k-5+o(1)$, we use same trick as in \cite{TZ01-spaa}. Each vertex $u\in A_0\setminus A_1$ will store in its routing table all the labels for vertices in $C(u)$, which enables to save an additive term of $d_G(u,v)$ in both $x_i$ and $y_i$. We refer the reader to \cite{TZ01-spaa} for the details.

\paragraph{Running time.}
By \theoremref{thm:app-pivot-cluster}, the time required to compute the approximate pivots and the trees $\tilde{C}(u)$ for every $u\in A_i\setminus A_{i+1}$ is $(n^{1/2+1/k}+D)\cdot \min\{(\log n)^{O(k)},2^{\tilde{O}(\sqrt{\log n})}\}$, when $k$ is even, and
$(n^{1/2+1/(2k)}+D)\cdot \min\{(\log n)^{O(k)},2^{\tilde{O}(\sqrt{\log n})}\}$, when $k$ is odd. By \claimref{claim:number-of-clusters},  each vertex participates in at most $\tilde{O}(n^{1/k})$ trees.  Hence, by  \remarkref{rem:n-trees},
which will be stated and proven in Section \ref{sec:tree},
  it will take only $\tilde{O}(n^{1/2+1/(2k)}+D)$ rounds to compute the routing tables for all trees in parallel. We conclude that the total number of rounds is $(n^{1/2+1/k}+D)\cdot \min\{(\log n)^{O(k)},2^{\tilde{O}(\sqrt{\log n})}\}$, for even $k$,
and $(n^{1/2+1/(2k)}+D)\cdot \min\{(\log n)^{O(k)},2^{\tilde{O}(\sqrt{\log n})}\}$, for odd.

\section{Distance Estimation}\label{sec:sketch}

In this section we sketch how the routing tables can be used for distance estimation, and prove the following.
\begin{theorem}\label{thm:sketch}
Let $G=(V,E)$ be a weighted graph with $n$ vertices and hop-diameter $D$, and let $k\ge 1$ be a parameter. Then there exists a distance estimation scheme, that assigns a sketch of size $O(n^{1/k}\log n)$ for each node, and has stretch $2k-1+o(1)$, that can be computed by a randomized distributed algorithm within $(n^{1/2+1/k}+D)\cdot \min\{(\log n)^{O(k)},2^{\tilde{O}(\sqrt{\log n})}\}$ rounds (whp).
In the case of odd $k$, the running time can be decreased to $(n^{1/2 + 1/(2k)} + D) \cdot \min\{(\log n)^{O(k)},2^{\tilde{O}(\sqrt{\log n})}\}$.
Furthermore, the distance computation can be done in time $O(k)$.
\end{theorem}

Apply \theoremref{thm:app-pivot-cluster}, which computes all the approximate pivots and approximate clusters. Each vertex $v\in V$ include in its sketch for every $u\in V$ so that $v\in\tilde{C}(u)$, the pair $(u,b_v(u))$, where $b_v(u)$ is the approximate distance to $u$ computed in \sectionref{sec:route}. Also for every $0\le i\le k-1$, add $(\hat{z}_i(v),\hat{d}_i(v))$, which is the approximate $i$-pivot and distance to it. By \remarkref{rem:siize},  every sketch is of size $O(n^{1/k}\log n)$. The algorithm that computes a distance estimate given two sketches is similar to that of \cite{TZ01}. We state it formally in \algref{alg:sketch}.

\begin{algorithm}
\caption{$\texttt{Dist}(u,v)$}\label{alg:sketch}
\begin{algorithmic}[1]
\STATE $i\leftarrow 0$;
\STATE $w\leftarrow u$;
\WHILE {$v\notin\tilde{C}(w)$}
\STATE $i\leftarrow i+1$;
\STATE $(u,v)\leftarrow(v,u)$;
\STATE $w\leftarrow \hat{z}_i(u)$;
\ENDWHILE
\RETURN $\hat{d}_i(u)+b_v(w)$;
\end{algorithmic}
\end{algorithm}

Observe that the sketch contains all the relevant information for executing \algref{alg:sketch}. When the while loop terminates $v\in\tilde{C}(w)$, so it has the estimate $b_v(w)$, while $u$ stores the approximate distance $\hat{d}_i(u)$ to every one of its approximate pivots. The stretch analysis is a variant of the analysis of \cite{TZ01}, similar in spirit to that of \sectionref{sec:ana}. Roughly speaking, on the stretch $2k-1$ achieved by \cite{TZ01}, we pay a multiplicative factor of $(1+O(\epsilon))^k$ due to the fact that distances are approximated. However,  this boils down to an $o(1)$ additive term, since $\epsilon=\frac{1}{48k^4}$. We leave the details to the reader.

\section{Distributed Tree Routing}\label{sec:tree}

In this section we present a modification of the (exact) routing scheme of Thorup-Zwick for rooted trees, that can be implemented efficiently in a distributed manner. The price is that the size of the labels and tables increases by a factor of $\log n$, compared to what \cite{TZ01-spaa} achieved.
\begin{theorem}\label{thm:tree-route}
Fix a graph $G=(V,E)$ on $n$ vertices with hop-diameter $D$. For any tree $T$ which is a subgraph of $G$, there is a routing scheme with stretch 1, routing tables of size $O(\log n)$ and labels of size $O(\log^2n)$, that can be computed in a distributed manner within $\tilde{O}(\sqrt{n}+D)$ rounds.
\end{theorem}
\begin{remark}\label{rem:n-trees}
If we are given $n$ trees, each a sub-graph of $G=(V,E)$, so that each vertex $v\in V$ participates in at most $s$ trees, then routing schemes for all the trees can be computed in $\tilde{O}(\sqrt{n\cdot s}+D)$ rounds.
\end{remark}
Let us first recall briefly how (a simplified version of) the TZ scheme works. For every non-leaf vertex, define a {\em heavy child} as the child with the largest subtree. Run a Depth First Search (DFS) on the tree, each vertex $u$ receives an entry time $a_u$ and exit time $b_u$.
The routing table stored at each vertex $u$ consists of the name and port number of its parent $p(u)$ in the tree, the name (and port) of its heavy child $h(u)$, and the numbers $a_u,b_u$. The label of a vertex $u$ contains the number $a_u$ and additional $\lceil\log n\rceil$ words: consider the path $P$ from the root to $u$, for every vertex $w$ on this path such that its heavy child is not on $P$, we append to the label of $u$ the name of $w$ and the port number leading from $w$ to its child on $P$. The observation is that whenever the path does not use the heavy child, the size of the subtree shrinks by a factor of at least 2, so this can happen only $\lceil\log n\rceil$ times. In order to route from $u$ to $v$, every intermediate vertex $x$ does as follows: if $a_x=a_v$ we are done, if $a_v\notin(a_x,b_x)$, we know the DFS did not find $v$ in the subtree rooted at $x$, so $x$ sends the message to its parent, and if $a_v\in(a_x,b_x)$ then $v$ lies in the subtree of $x$. In the latter case, $x$ examines the label of $v$ for an entry of the form $(x,x')$, if it exists it sends to its child $x'$, if not, $x$ sends the message to its heavy child.

In order to obtain a scheme that runs efficiently in a distributed manner, we cannot compute heavy children and run DFS on the entire tree. Instead, we shall apply certain variants of the TZ-scheme in two levels. Let $T$ be a tree on the vertices $V(T)\subseteq V$, rooted at $z$. For $u\in V(T)$, denote by $p(u)$ the parent of $u$ in $T$. We assume that every vertex knows the names of its parent and its children. The basic idea is to randomly sample $\gamma \ge c \cdot  \ln n$, for a sufficiently large constant $c$, vertices $U\subseteq V$. ($\gamma$ here is a parameter.) Each vertex in $V$ chooses itself to $U$ independently with probability $\frac{\gamma}{n}$. Partition the tree $T$ into subtrees according to the vertices of $U(T)=(U\cap V(T))\cup \{z\}$, by removing each edge from a vertex of $U(T)$ to its parent.
 Note that this partitions $T$ into a forest $F$ of $|U(T)|$ subtrees, each of these subtrees is rooted at a vertex of $U(T)$. For $w\in U(T)$, denote by $T_w$ the subtree in $F$ rooted in $w$. Let $T'$ denote the virtual tree on the vertices of $U(T)$, where $w$ is a parent of $u$ in $T'$, if $p(u)$ lies in $T_w$. We shall devise a routing scheme for each $T_w$, and a global scheme that routes in $T'$. We begin by bounding the depth of each subtree; let $B={{4n} \over \gamma} \cdot \ln n$.
\begin{claim}\label{claim:tree-size}
With high probability,  $|U| = O(\gamma)$, and for each $w\in U(T)$, the tree $T_w$ has depth at most $B$.
\end{claim}
\begin{proof}
The first event holds with high probability by a simple Chernoff bound. For the second: by independence, the probability that a path $P$ in $T$ of length $B$ has $P\cap U=\emptyset$, is
\[
\left(1-\frac{\gamma}{n}\right)^{4n/\gamma \ln n}\le \frac{1}{n^4}~.
\]
Taking a union bound on the $O(n^2)$ possible paths (in a tree, choosing the path's endpoints determines it) completes the proof.
\end{proof}
\noindent{\bf Remark:} Observe that we still have high probability that the events of \claimref{claim:tree-size} hold over $n$ different trees of the Thorup-Zwick cover.

From now on assume the events of \claimref{claim:tree-size} hold.
The assignment has two phases.
\paragraph{Phase 1.} In the first phase we compute a routing scheme for each $T_w$ in the forest $F$, in parallel. In each round, every vertex $u$ that received messages from all its children, sends to its parent in $F$ the size of its subtree (by summing up the sizes of the subtrees of the children of $u$). By \claimref{claim:tree-size}, the depth of each tree in $F$ is at most $B$, and in each round we send one word per vertex. Hence after $B$ rounds every vertex knows the size of its subtree (in $F$), and in particular, can infer who is its heavy child.
Now each $w\in U(T)$ can start a parallel DFS of $T_w$ -- that is, every vertex assigns entry and exit times to all if its children in parallel (it is possible since it knows the sizes of every child's subtree). Each vertex in $T_w$ adds to its routing table $(p(x),h(x),a_x,b_x,w)$, which are the name of the parent of $x$, the heavy child of $x$, the entry and exit times, and the name $w$. This computation (parallel DFS) will also require $O(B)$ rounds, since all subtrees work in parallel.

The (local) label assignment for vertices in $T_w$ is done in the following manner. Starting from $w$ (which has empty label), every vertex $x$ that receives a label $\ell$ from its parent, and has children $x_1,\dots,x_l$, sends $\ell$ to its heavy child, and $\ell\circ(x,x_i)$ to $x_i$ for each non-heavy child $x_i$. The label $\ell(x)$ will consist of $a_x$ and the list $\ell$ of edges that was given to $x$.

\paragraph{Phase 2.} In the second phase we compute a routing scheme on $T'$. Every $u\in U(T)$ sends a message to its parent $x$ in $T$, and receives from $x$ the following message: $\ell(x)$, the name $w$ such that $x\in T_w$ (so that the edge $(w,u)$ should be in $T'$), and also the port number $e(x,u)$ of $x$ leading to $u$. Then every such $u$ broadcasts $((w,u),x,\ell(x), e(x,u))$ to the entire graph. Once the root vertex  $z$ has full information on $T'$, it may locally compute the TZ routing scheme for $T'$. The routing table given to $u\in U(T)$ is slightly different than in the usual scheme, as it will contain local routing information for the vertex leading to the heavy child. More formally, the table will be $(h'(u),\ell(y),e(y,h'(u)),a'_u,b'_u)$. Here $h'(u)$ is the name of heavy child of $u$ in $T'$, $y\in T_u$ is the portal vertex which is the parent of $h'(u)$ in $T$, and $e(y,h'(u))$ is the port of $y$ leading to $h'(u)$. Note that $z$ has the name, label and the appropriate port of $y$ when $h'(u)$ reported the edge $(u,h'(u))$. Finally $a'_u,b'_u$ are the entry and exit times of the DFS run by $z$ on $T'$.
Observe that $\ell(y)$ has size $O(\log n)$, and this term dominates the size of a routing table.
 There are at most $O(\gamma)$ such tables. Hence \lemmaref{lem:pipe} implies  that we can broadcast to the entire graph all these messages within $O(\gamma\log n+D)$ rounds. In addition, every vertex $u\in U(T)$ sends the routing table given to it to all the vertices in $T_u$. Since we can send the information inside each subtree in parallel, it will take only $O(B\log n)$ rounds.

The label assignment to the vertices of $T'$ is also modified, since for every possible edge taken in $T'$ which is not leading to a heavy child, we must add the local routing information. Fix $u\in U(T)$. Assume $((v_1,w_1),\dots,(v_l,w_l))$ is the list of all edges in the path of $T'$ from $z$ to $u$, so that each $w_i$ is a non-heavy child of $v_i$. Ordinarily, this list would have been the label of $u$ (along with $a'_u$). However, in order to be able to route in $T'$, we replace each such edge with $(v_i,w_i,\ell(x_i),e(x_i,w_i))$, where $x_i$ is the parent of $w_i$ in $T$, $\ell(x_i)$ is the label $x_i$ received in the first phase (for local routing within $T_{v_i}$), and $e(x_i,w_i)$ is the port leading from $x_i$ to $w_i$. Recall that $z$ knows the label and appropriate port of every such $x_i$. Since each $\ell(x_i)$ has size at most $O(\log n)$ words, and $l\le \log n$, we have that the label size is $O(\log^2n)$. As before, each $u\in U(T)$ propagates this label $\ell'(u)$ to every vertex in $T_u$. The number of rounds is therefore $O(\gamma \log^2n+D)$.

\paragraph{Protocol.}
The routing from $u$ to $v$ will be done as follows. Assume we have arrived to an intermediate vertex $x$ that lies in $T_w$. First $x$ checks if routing in $T'$ is required, by comparing $a'_v$ with $a'_x,b'_x$ (recall that $a'_v$ is part of the label of $v$, and the routing table of $x$ contains $a'_x=a'_w$ and $b'_x=b'_w$). If $a'_v=a'_x$ then $v\in T_w$, and we proceed to route inside $T_w$. If $a'_v\notin(a'_x,b'_x)$, we need to route to the subtree rooted at the parent of $w$ in $T'$, and if $a'_v\in(a'_x,b'_x)$ then we need to route to the appropriate child of $w$ in $T'$,

{\bf Routing inside $T_w$:} This is done exactly as in the TZ scheme, while considering the local routing tables of vertices in $T_w$ and $\ell(v)$. If $a_x=a_v$ we are done. If $a_v\notin(a_x,b_x)$ we route to the parent of $x$ (stored in the local routing table of $x$), and when $a_v\in(a_x,b_x)$, we inspect $\ell(v)$: if it contains an edge of the form $(x,x')$, for some $x'$, we route to $x'$. Otherwise to the heavy child of $x$ (the heavy child's name is also in the local routing table of $x$).

{\bf Routing to the parent of $w$ in $T'$:} This is simple, $x$ just routes to its parent, its name is stored in the local routing table of $x$. Eventually we will reach $w$ (since all vertices in $T_w$ have the same $\ell'$ label), and route from it to vertex in the tree of $w$'s parent in $T'$.

{\bf Routing to a child of $w$ in $T'$:} Here we inspect $\ell'(v)$, if it contains an entry of the form $(w,w',\ell(y),e(y,w'))$ then we know we have to route in $T'$ from $w$ to its child $w'$ in $T'$. Fortunately, the label $\ell(y)$ provides us the required routing information to route in $T_w$ to the portal vertex $y$ (that has $w'$ as a child in $T$). From $y$ we go to its child $w'$ using the port $e(y,w')$. If the label $\ell'(v)$ contains no such entry, then we know we need to route to the heavy child of $w$ in $T'$. Here the label of $v$ is useless, but we stored the label of $y'\in T_w$, the portal vertex which is the parent of this heavy child, in the routing table of each vertex of $T_w$. Using the label of $y'$ we can route locally in $T_w$, and from $y'$ route to $h'(w)$ (using the port number for heavy child stored in the routing table).

When constructing routing tables and labels for one single tree, the overall running time is
$O(\gamma \cdot \log^2 n + D) + O(B \cdot \log n) = O(\gamma \cdot \log^2 n + {n \over \gamma} \cdot \log^2 n + D)$, i.e.,
$O(D+ \sqrt{n} \cdot \log^2 n)$, by setting $\gamma = \sqrt{n}$.

\begin{proof}[Proof of \remarkref{rem:n-trees}]
To avoid high running time, we shall perform the routing tables and labels computations in parallel in all  cluster trees, while appending to each message the name of the relevant tree. In the first phase, which can be implemented in $\tilde{O}(\sqrt{n})$ rounds for each tree, we send information on the graph edges (every vertex notifies all its neighbors in each round), so the overhead due to participation in up to $s$ trees is only a factor of $s$. In the second phase, however, we broadcast messages to the entire graph.  So we need a bound on the number of these messages. For each tree $T'$ (which consists of the vertices of $U$ alone) we broadcast 2 messages per vertex: the first informing the root of its existence, its parent, and the local routing information. In the second message, the root broadcasts routing information and a label for the vertex. Each message is of size $O(\log^2n)$. By charging these messages to the vertices of $U$, each such vertex pays for 2 messages per tree containing it. But the number of these trees is at most $s$, so we need to broadcast at most $\tilde{O}(\sqrt{n}\cdot s)$ words. By \lemmaref{lem:pipe}, these can be broadcast to the entire graph in $\tilde{O}(\sqrt{n}\cdot s+D)$ rounds.

We next argue that this bound can be further improved to $\tO(\sqrt{n \cdot s} +D)$.

Every root $w$ of a tree $T_w$ in one of the forests $F$ (each cluster tree gives rise to one such a forest) tosses a starting time $\st(w)$
uniformly at random from the interval $[1,c \cdot \ln n \cdot \sqrt{ns}]$, for a sufficiently large constant $c$.
It then starts broadcasting to vertices of $T_w$ at time $20 \cdot \st(w)$. (It broadcasts to them the value $\st(w)$.)
Each round of this broadcast is replaced by stages consisting of 20 rounds each. Specifically, a vertex $x$ in $T_w$ that already received the message from its parent tries to deliver it to its children for 20 consecutive rounds. We will show that, whp, for every edge, on one of these rounds no congestion will be experienced. Only when these 20 rounds are over, the children of $x$ will start broadcasting.

Consider a specific  edge $e = (x,y)$ in a tree $T_w$. Let $w_1,w_2,\ldots,w_s$ be the roots of trees $T_{w_i}$ that contain this edge.
(Recall that, by Claim \ref{claim:number-of-clusters}, whp, $s = O(n^{1/k} \cdot \log n)$.)  Let $t_1,t_2,\ldots,t_s$ be the respective hop-distances between $w_i$ and the closer endpoint of $e_i$ to $w_i$. In other words, for every $i \in [s]$, if $w_i$ broadcasted a message over $T_{w_i}$, and no other messages would have interfered with its broadcast, then the broadcast of $w_i$ would traverse $e_i$ on step $t_i$. (For convenience, we number the steps starting from 0.)

For any index  $R$, the probability that the broadcast of $w_i$ will want to traverse $e$ on stage $R$, conditioned on the assumption that it experienced no congestion whatsoever before that, is the probability that $w_i$ starts broadcasting at stage $R-t_i$, i.e., this is equal to
$\Prob(\st(w_i) = R - t_i)$.
The latter probability is at most ${1 \over {c \sqrt{n  s} \ln n}}$. For a positive integer $\alpha \le s$, the probability that $\alpha$ cluster trees wish to employ $e$ on stage $R$, conditioned on the assumption that no congestion was experienced by any of them so far, is at most
$$\left({1 \over {c \ln n \cdot \sqrt{n s}}}\right)^\alpha \cdot {s \choose \alpha} ~\le ~
\left({s \over {c \ln n \cdot \sqrt{n s}}}\right)^\alpha ~\le~ \left({1 \over {n^{1/2 - 1/(2k)}}}\right)^\alpha~.$$
For $\alpha = 20$, this probability is at most ${1 \over {n^{10 - 10/k}}} \le {1 \over {n^5}}$.
By union-bound over all stage indices $R \le n$, and all the $|E| \le n^2$ edges, we still have an only negligible probability that a congestion was ever experienced throughout the algorithm. (Here we say that a congestion is experienced if a vertex $v$ wishes to broadcast a message $m$ on a stage $R$ of the algorithm through an edge $(v,u)$ incident on $v$, and $v$ cannot do it for the entire $\alpha = 20$ rounds of this stage, because of other transmissions that employ the same edge.)

Hence, whp, in $O(B \cdot \alpha) + O(\sqrt{n s} \ln n) = \tO(B + n^{1/2 + 1/(2k)}\ln n)$ rounds, all broadcasts of the values of starting times will be completed. (Recall that $B$ is an upper bound on the depth of trees $T_{w_i}$.)
This completes Phase 0 of the algorithm.

Now the algorithm proceeds to Phase 1, on which convergecasts are conducted in all these trees. As a result of these convergecasts, every vertex $x \in T_{w_i}$ knows the size of its subtree in $T_{w_i}$.
These convergecasts are conducted by a similar procedure to the one that was described above, i.e., all leaves of $T_{w_i}$ start broadcasting at stage $\st(w_i)$, and each stage lasts for $\alpha = 20$ rounds. Hence these convergecasts are also completed in $O(B + \sqrt{n s} \cdot \ln n)$ rounds.
Then the ``parallel DFSs'' are conducted in all the trees in parallel by the same procedure of tree broadcast. As a result, all vertices $x$ in these trees $T_{w_i}$ learn their routing tables within $T_{w_i}$. They also learn their routing labels within additional $O(B \log n + \sqrt{n s} \log^2 n)$ time.
(Note that for labels one may need to send messages of size $O(\log n)$ words, and so stages of length $O(\alpha \cdot \log n) = O(\log n)$ are needed.)

Phase 2 is performed in the same way as was already described.   Specifically, the algorithm conducts convergecasts of messages $(\ell(x),w,e(x,u))$, where $u \in U(T)$ and $x$ is its parent in $T$, for some cluster tree $T$, over the BFS tree $\tau$ of the entire graph $G$. Since every selected vertex $u$ may participate in up to $s$ trees, and there are $O(\gamma)$ selected vertices, this convergecast requires $O(\gamma \cdot s + D)$ time. Analogously, the broadcast of the computed routing tables requires $O(\gamma \cdot s \log n + D)$ time.

Then each $u \in U(T)$ sends its routing table to all vertices of $T_u$. This is done using the tossed starting times and with stages of $\alpha$ rounds each, as in Phase 1. Hence this step requires $O(B \log n + \sqrt{n s} \log^2 n)$ time. Finally, the labels of selected nodes in $T'$ are broadcasted over the BFS tree $\tau$ within additional $O(\gamma \cdot s \cdot \log^2 n + D)$ time.

To summarize, the overall running time of the algorithm is $\tO(B   + D + \sqrt{n s} + \gamma \cdot s) = \tO({n \over \gamma }+ D + n^{1/2 + 1/(2k)}  + \gamma \cdot s)$.
By setting $\gamma = \sqrt{n/s} = {{n^{1/2 - 1/(2k)}} \over {\sqrt{\log n}}}$, we get the running time of
$\tO(\sqrt{n s} +D) = \tO(n^{1/2 + 1/(2k)}   +D)$.

\end{proof}

\bibliographystyle{alpha}
\bibliography{route}

\appendix

\section{Proof of \theoremref{thm:SPT}}\label{app:SPT}

Let $X\subseteq V$ be a set of vertices so that each $v\in V$ is sampled to $X$ independently with probability $1/\sqrt{n}$. Define $V'=A\cup X$, and note that with high probability $ B=4\sqrt{n}\ln n\ge |V'|$ (since it is given that $|A|\le 2\sqrt{n}\ln n$).
Apply the same preprocessing steps as in \sectionref{sec:pre} with $V'$ as defined here, to obtain a graph $G''$ on $V'$ satisfying \eqref{eq:g''}.

\paragraph{Computing Approximate SPT for $V'$.}
The first step is to compute the values $(\hat{d}(v),\hat{z}(v))$ for vertices $v\in V'$.
Every vertex in $v\in A$ initializes its values as $(0,v)$, while $v\notin A$ sets $(\infty,\bot)$.
Conduct $\beta=\min\{2^{\tilde{O}(\sqrt{\log n})},(\log n)^{O(k)}\}$ iterations of Bellman-Ford rooted at $A$: at every iteration, every vertex $v\in V'$ broadcasts its pair $(\hat{d}(v),\hat{z}(v))$ to the entire graph, and if $u\in V'$ has $w''(u,v)+\hat{d}(v)<\hat{d}(u)$, then $u$ updates its pair to be $(w''(u,v)+\hat{d}(v),\hat{z}(v))$. (Recall that $w''$ is the edge weight function of $G''$, where the latter is the virtual graph given by \theoremref{thm:Nanongkai} augmented with the hopset edges of \theoremref{thm:hopset}.)

The number of rounds required to construct $G''$ is $(n^{1/2+1/(2k)}+D)\cdot\min\{2^{\tilde{O}(\sqrt{\log n})},(\log n)^{O(k)}\}$, and by \lemmaref{lem:pipe} this term also bounds the number of rounds it takes to broadcast the $O(|V'|\cdot\beta)$ messages for the Bellman-Ford iterations.

\paragraph{Extending the SPT to $V$.}
At the end of the $\beta$ iterations of Bellman-Ford, every vertex $u\in V$ knows $(\hat{d}(v),\hat{z}(v))$ for every $v\in V'$. Every vertex $u\in V$ computes
\begin{equation}\label{eq:defd}
\hat{d}(u)=\min_{v\in V'}\{d_{uv}+\hat{d}(v)\}~,
\end{equation}
and sets $\hat{z}(u)=\hat{z}(v)$, where $v\in V'$ is the minimizer of \eqref{eq:defd}. (Recall that $d_{uv}$ is the value computed in \theoremref{thm:Nanongkai}.)

\paragraph{Analysis.}

We assume all the events of \claimref{claim:hit-path} hold (which happens with high probability).
For $u\in V$ let $z_u\in A$ be a vertex satisfying $d_G(u,z_u)=d_G(u,A)$.
Since we performed $\beta$ iterations of Bellman-Ford, using \eqref{eq:g''} with $v\in V'$ and $z_v\in A\subseteq V'$ we have that $v'$ satisfies \eqref{eq:stretc}.

Consider now some $u\in V$, and let $v\in V'$ be the minimizer in \eqref{eq:defd}. The left hand side of \eqref{eq:stretc} holds, as the fact that $v\in V'$ satisfies \eqref{eq:stretc} implies
\[
d_{uv}+\hat{d}(v)\stackrel{\eqref{eq:duv}}{\ge} d_G^{(B)}(u,v)+d_G(v,A)\ge d_G(u,v)+d_G(v,A)\ge d(u,A)~.
\]
For the right hand side of \eqref{eq:stretc}: In the case that $h(u,z_u)\le B$, by \eqref{eq:duv} we get that
\[
\hat{d}(u)\le d_{uz_u}+\hat{d}(z_u)\le (1+\epsilon)d_G^{(B)}(u,z_u)+0=(1+\epsilon)d_G(u,z_u)~.
\]
Otherwise $h(u,z_u)> B$, and by \claimref{claim:hit-path} there exists $v\in X\subseteq V'$ on the shortest path in $G$ from $u$ to $z_u$ with $h(u,v)\le B$. Since \eqref{eq:stretc} holds for $v$,
\begin{eqnarray*}
\hat{d}(u)&\le& d_{uv}+\hat{d}(v)\\
&\stackrel{\eqref{eq:duv}}{\le}& (1+\epsilon)d_G^{(B)}(u,v)+ (1+\epsilon)d_G(v,A)\\
&\le& (1+\epsilon)d_G(u,v)+ (1+\epsilon)d_G(v,z_u)\\
&=&(1+\epsilon)d_G(u,z_u)~.
\end{eqnarray*}

\end{document}